\documentclass[conference]{IEEEtran}
\IEEEoverridecommandlockouts
\usepackage{amsmath,amssymb,amsfonts,amsthm}
\usepackage{algorithmic}
\usepackage{graphicx}
\usepackage{float}
\usepackage{textcomp}
\usepackage{xcolor}
\def\BibTeX{{\rm B\kern-.05em{\sc i\kern-.025em b}\kern-.08em
    T\kern-.1667em\lower.7ex\hbox{E}\kern-.125emX}}
\usepackage[utf8]{inputenc}
\usepackage{color}
\usepackage[skins]{tcolorbox}
\usepackage{tikz}
\usetikzlibrary{backgrounds}
\usepackage{hyperref}
\usepackage{graphicx}
\usepackage{amsmath}
\usepackage{xparse}
\usepackage{xspace}
\usepackage{pgfplots}
\usepgfplotslibrary{dateplot}
\usepackage[final]{pdfpages}
\usepackage{framed}

\usepackage{filecontents}
\usepackage{pgfplotstable}
\usepgfplotslibrary{statistics}

\usepackage[n,advantage,operators,keys,events,notions,complexity, primitives,adversary,landau,ff,sets,probability]{cryptocode}

\newtheorem{example}{Example}

\definecolor{tylergreen}{HTML}{194722}

\newcommand{\frme}{\href{http://frontrun.me/}{frontrun.me}}

\usepackage[backend=bibtex,
style=numeric,
    bibencoding=ascii
]{biblatex}
\begin{document}

\title{Flash Boys 2.0: \\ Frontrunning, Transaction Reordering, and Consensus Instability in Decentralized Exchanges}
\author{\centering
\IEEEauthorblockN{}
\IEEEauthorblockA{\phantom{Corne} \\
}
\and
\IEEEauthorblockN{Philip Daian}
\IEEEauthorblockA{\textit{Cornell Tech} \\
phil@cs.cornell.edu}
\and
\IEEEauthorblockN{Steven Goldfeder}
\IEEEauthorblockA{\textit{Cornell Tech} \\
goldfeder@cornell.edu}
\and
\IEEEauthorblockN{Tyler Kell}
\IEEEauthorblockA{\textit{Cornell Tech} \\
sk3259@cornell.edu}
\and
\IEEEauthorblockN{Yunqi Li}
\IEEEauthorblockA{\textit{UIUC} \\
yunqil3@illinois.edu}
\and
\IEEEauthorblockN{Xueyuan Zhao}
\IEEEauthorblockA{\textit{CMU} \\
xyzhao@cmu.edu}
\and
\IEEEauthorblockN{}
\IEEEauthorblockA{\phantom{Corne} \\
}
\and
\IEEEauthorblockN{}
\IEEEauthorblockA{\phantom{Cornell Tech} \\
}
\and
\IEEEauthorblockN{Iddo Bentov}
\IEEEauthorblockA{\textit{Cornell Tech} \\
ib327@cornell.edu}
\and
\IEEEauthorblockN{Lorenz Breidenbach}
\IEEEauthorblockA{\textit{ETH Zürich} \\
lorenz.breidenbach@inf.ethz.ch}
\and
\IEEEauthorblockN{Ari Juels}
\IEEEauthorblockA{\textit{Cornell Tech} \\
juels@cornell.edu}
\and
\IEEEauthorblockN{}
\IEEEauthorblockA{\phantom{C} \\
}
}
\maketitle
\thispagestyle{plain}
\pagestyle{plain}

\begin{abstract}

Blockchains, and specifically smart contracts, have promised to create fair and transparent trading ecosystems.

Unfortunately, we show that this promise has not been met. We document and quantify the widespread and rising deployment of {\em arbitrage bots} in blockchain systems, specifically in \emph{decentralized exchanges} (or ``DEXes"). Like high-frequency traders on Wall Street, these bots exploit inefficiencies in DEXes, paying high transaction fees and optimizing network latency to frontrun, i.e., anticipate and exploit, ordinary users' DEX trades.  

We study the breadth of DEX arbitrage bots in a subset of transactions that yield quantifiable revenue to these bots. We also study bots' profit-making strategies, with a focus on blockchain-specific elements. We observe bots engage in what we call \emph{priority gas auctions} (PGAs), competitively bidding up transaction fees in order to obtain priority ordering, i.e., early block position and execution, for their transactions. PGAs present an interesting and complex new continuous-time, partial-information, game-theoretic model that we formalize and study. We release an interactive web portal, \frme, to provide the community with real-time data on PGAs.

We additionally show that high fees paid for priority transaction ordering poses a systemic risk to {\em consensus-layer} security. We explain that such fees are just one form of a general phenomenon in DEXes and beyond---what we call \emph{miner extractable value} (MEV)---that poses concrete, measurable, consensus-layer security risks. We show empirically that MEV poses a realistic threat to Ethereum today.

Our work highlights the large, complex risks created by transaction-ordering dependencies in smart contracts and the ways in which traditional forms of financial-market exploitation are adapting to and penetrating blockchain economies.

\end{abstract}
\newcommand{\newType}[3]{%
  \DeclareDocumentCommand{#1}{oo}{%
    \ensuremath{%
      \textcolor{#2}{%
        % NOTE: Uncomment this stuff if we need to compress lines. It will
        % make the spacing algorithm ignore the super/subscripts by treating
        % them as height 0. It might make things a little harder to read, but
        % we'll get more space.
        \mathchoice%
          {{#3}\IfNoValueF{##1}{_{##1}}\IfNoValueF{##2}{^{##2}}}%
          {{#3}%
            \IfNoValueF{##1}{_{\raisebox{0pt}[0pt]{\scriptsize\ensuremath{##1}}}}%
            \IfNoValueF{##2}{^{\raisebox{0pt}[0pt]{\scriptsize\ensuremath{##2}}}}}%
          {{#3}\IfNoValueF{##1}{_{##1}}\IfNoValueF{##2}{^{##2}}}%
          {{#3}\IfNoValueF{##1}{_{##1}}\IfNoValueF{##2}{^{##2}}}%
      }%
    }\xspace%
  }%
  \WithSuffix\newcommand#1'[1][]{\ensuremath{\textcolor{#2}{{#3}'_{##1}}}\xspace}%
  \WithSuffix\newcommand#1*[1][]{#1[##1][\$]}%
}
\definecolor{DarkGreen}{rgb}{0.0, 0.4, 0.0}
\colorlet{money}{DarkGreen}

%%%%%%% Auction model notation %%%%%%%%

\newcommand{\strategy}{\ensuremath{S}}
\newcommand{\currenttime}[1]{\ensuremath{{t_{#1}^*}}}
\newcommand{\timet}{\ensuremath{t}}
\newcommand{\waketime}[1]{\ensuremath{\timet_{w#1}}}
\newcommand{\move}{\ensuremath{a}}
\newcommand{\bidtuple}{\ensuremath{b}}
\newcommand{\bidtuples}{\ensuremath{{\bf b}}}
\newcommand{\currentbidtuples}{\ensuremath{{\bf b}^*}}
\newcommand{\bid}{\ensuremath{b}}
\newcommand{\playstate}{\ensuremath{\sigma}}
\newcommand{\latency}{\ensuremath{\Delta}}
\newcommand{\wake}{\ensuremath{\delta}}
\newcommand{\player}{\ensuremath{\mathcal{P}}}
\newcommand{\pindex}{\ensuremath{i}}
\newcommand{\exec}{\ensuremath{{\sf Exec}}}
\newcommand{\pend}{\ensuremath{{\bf p}}}
\newcommand{\delayed}[1]{\ensuremath{{\textbf{d}_{#1}}}}
\newcommand{\firstbidtime}{\ensuremath{{\sf FirstBidTime}}}
\newcommand{\firsttime}{\ensuremath{{\sf FirstTime}}}
\newcommand{\bidplayer}{\ensuremath{{\sf BidPlayer}}}
\newcommand{\bidtime}{\ensuremath{{\sf BidTime}}}
\newcommand{\nexttime}{\ensuremath{{\sf NextBidTime}}}
\newcommand{\firstbid}{\ensuremath{{\sf FirstBid}}}
\newcommand{\popfirst}{\ensuremath{{\sf PopFirstBid}}}
\newcommand{\poptime}{\ensuremath{{\sf PopFirstTime}}}
\newcommand{\maxbid}{\ensuremath{{\sf MaxBid}}}
\newcommand{\blockinterval}{\ensuremath{D}}
\newcommand{\interval}{\ensuremath{\delta}}
\newcommand{\payoffexec}[2]{\ensuremath{{\sf PO}_{#2}^{#1}}}

\newcommand{\lossmult}{\ensuremath{\ell}}

\newType{\payoff}{money}{\$\mathsf{1}}
\newType{\bidcost}{money}{\$\mathsf{c}}
\newType{\payoffa}{money}{\$\mathsf{r}}
\newType{\bid}{money}{\$\mathsf{b}}
\newType{\gain}{money}{\$\mathsf{r}}

\newcommand{\specialcell}[2][c]{%
  \begin{tabular}[#1]{@{}l@{}}#2\end{tabular}}

\newtheorem{observation}{Observation}
\newtheorem{theorem}{Theorem}
\newtheorem{lemma}{Lemma}

\newtcolorbox{myframe}[2][]{%
  enhanced,colback=white,colframe=black,coltitle=black,
  sharp corners,boxrule=0.4pt,
  fonttitle=\itshape,
  attach boxed title to top left={yshift=-0.3\baselineskip-0.4pt,xshift=2mm},
  boxed title style={tile,size=minimal,left=0.5mm,right=0.5mm,
    colback=white,before upper=\strut},
  title=#2,#1
}
\section{Introduction}
\label{sec:intro}

Cryptocurrency exchanges today handle more than \$10 billion in trade volume per day. The vast majority of this volume occurs in {\em centralized} exchanges, which hold custody of customer assets and settle trades. At best loosely regulated, centralized exchanges have experienced scandals ranging from high-profile thefts~\cite{mcmillan2014inside} to malfeasance such as price manipulation~\cite{gandal2018price}. One popular alternative is what is called a \emph{decentralized exchange} (or ``DEXes'').\footnote{``Decentralized" exchange is something of a misnomer, as many such systems have centralized components; most systems we call ``decentralized" exchanges could more accurately be classified as non-custodial: users trade without surrendering control of their funds to a third party in the process.}  In a DEX, a smart contract (a program executing on a blockchain) or other form of peer-to-peer network executes exchange functionality.

At first glance, decentralized exchanges seem  ideally designed. They appear to provide effective price discovery and fair trading, while doing away with the drawbacks of centralized exchanges. Trades are atomically executed by a smart contract and visible on the Ethereum blockchain, providing the appearance of transparency. Funds cannot be stolen by the exchange operator, because their custody and exchange logic is processed and guaranteed by the smart contract.

Despite their clear benefits, however, many DEXes come with a serious and fundamental weakness: on-chain, smart-contract-mediated trades are slow.\footnote{The average Ethereum block time is roughly 15s at the date of writing~\cite{Etherscan}.} Traders thus may attempt to take orders that have already been taken or canceled but appear active due to their views of messages sent on the network. Worse still, adversaries can {\em frontrun} orders, observing them and placing their own orders with higher fees to ensure they are mined first.

Past work has acknowledged ``transaction ordering dependence'' as an anti-pattern and vector for potential frontrunning~\cite{kosba2016hawk,luu2016making}.  Unfortunately, these analyses have previously proved overly broad: virtually every smart contract can be said to have \emph{some} potential dependence on transaction order, the majority of which is benign.  As a result, effective practical mitigations for these issues have failed to materialize, and few deployed smart contracts feature ordering protections.  Other work has focused on systematizing knowledge around smart contract frontrunning~\cite{eskandari2019sok}, including citing early public versions of this work, but has not measured the size of this economy or formalized its connection to protocol attacks.

In this work, we explain that DEX design flaws threaten underlying blockchain security. We study a community of {\em arbitrage bots} that has arisen to exploit DEX flaws. We show that these bots exhibit many similar market-exploiting behaviors---frontrunning, aggressive latency optimization, etc.---common on Wall Street, as revealed in the popular Michael Lewis expos\'{e} {\em Flash Boys}~\cite{lewis2014flash}. We explore the DEX design flaws that spawned arbitrage bots, measure and model these bots' behavior, and illuminate systemic smart-contract ecosystem risks implied by our observations. Our main focuses are:

\vspace{1mm}
\noindent \textbf{Pure revenue opportunities}: A specific sub-category of DEX arbitrage representative of broader activity, these are blockchain transactions that issue multiple trades atomically through a smart contract and profit unconditionally in every traded asset. We choose these opportunities as a focus because their simplicity makes them especially amenable to study and measurement. We experimentally determine a lower bound on this economy of over USD 6M to date and describe its participating exchanges and bots.

\vspace{1mm}
\noindent \textbf{Priority gas auctions (PGAs)}: Because pure revenue opportunities offer unconditional revenue, arbitrage bots often compete against each other by bidding up transaction fees (gas) in what we call PGAs.  We formally model bot PGA behavior and observe a cooperative equilibrium. We show that empirical measurements of the evolution of bot PGA strategies validate key features of our model.

\vspace{1mm}
\noindent \textbf{Miner-extractable value (MEV)}: We introduce the notion of MEV, value that is extractable by miners directly from smart contracts as cryptocurrency profits.  One particular source of MEV is \emph{ordering optimization (OO) fees}, which result from a miner's control of the ordering of transactions in a particular epoch. PGAs and pure revenue opportunities provide one source of OO fees. We show that MEV creates systemic consensus-layer vulnerabilities.

\vspace{1mm}
\noindent\textbf{Fee-based forking attacks}: We show that OO fees can incentivize miners to mount forking attacks. While fee-based attacks were previously studied theoretically in Bitcoin \cite{carlsten2016instability}, we empirically demonstrate a current, realistic threat in Ethereum.

\vspace{1mm}
\noindent\textbf{Time-bandit attacks}: We show that high-MEV regimes in general lead to a new attack in which miners rewrite blockchain history to steal funds allocated by smart contracts in the past. We call these {\em time-bandit attacks}. Our experiments show that MEV from pure revenue profits and PGA bot fees suffice to enable time-bandit attacks {\em on today's Ethereum}.  

\vspace{2mm}

Our results are surprising for two key reasons.

\vspace{2mm}

First, they identify a concrete difference between the consensus-layer security model required for blockchain protocols securing simple payments and those securing smart contracts. In a payment system such as Bitcoin, all independent transactions in a block can be seen as executing atomically, making ordering generally unprofitable to manipulate.  Our work shows that analyses of Bitcoin miner economics fail to extend to smart contract systems like Ethereum, and may even require modification once second-layer smart contract systems that depend on Bitcoin miners go live~\cite{lerner2015rootstock}.

Second, our analysis of PGA games underscores that protocol details (such as miner selection criteria, P2P network composition, and more) can directly impact application-layer security and the fairness properties that smart contracts offer users. Smart contract security is often studied purely at the application layer, abstracting away low-level details like miner selection and P2P relayers' behavior in order to make analysis tractable (e.g.~\cite{buterin2014schellingcoin,gyges, zhangparalysis,zhang2016town}). Our work shows that serious blind spots result. Low-level protocol behaviors pose fundamental challenges to developing robust smart contracts that protect users against exploitation by profit-maximizing miners and P2P relayers that may game contracts to subsidize attacks.

\vspace{2mm}

To illuminate the behaviors explored in this paper, we release a web dashboard, \frme, that presents PGA data in real time. We open-source all associated code and hundreds of gigabytes of raw data on the Ethereum PGA economy (with processed data).\footnote{Our Github repository at \url{https://github.com/pdaian/flashboys2} contains all infrastructure, data processing, and visualization code and data, plus code for our original arbitrage trade bot.}

We hope our efforts in general offer  insight into the broad, application- and consensus-layer risks created by ordering dependencies in smart contracts and into the effects of traditional financial-market exploitation on blockchain consensus.
\section{Background}

We now provide background required to understand PGAs. 

\subsection{Smart Contracts}

Smart contracts are small computer programs executed without user intervention, often by a system that allows all of its participants to verify these programs' correct execution. Smart contracts often use a public blockchain network as the underlying infrastructure for their execution~\cite{szabo1997formalizing,wood2017ethereum}.  

Ethereum~\cite{wood2017ethereum} is currently the largest smart contract system that is Turing-complete, i.e., allows encoding of arbitrary smart-contract functionality. Ethereum smart contracts have been used or proposed for a range of complex transaction types, including shareholder voting~\cite{mccorry2017smart}, stakeholder-owned investment funds and vehicles~\cite{DAO,dao-attack}, fair exchange protocols for goods~\cite{zhang2016town}, complex key management solutions~\cite{zhangparalysis}, video games~\cite{kharifcryptokitties}, virtual casinos~\cite{meng2018understanding}, and more.

The most popular smart contracts on Ethereum by daily active users primarily concern virtual sub-currencies called {\em tokens}.  These tokens can represent any scarce item, e.g. collectible resources in a video game~\cite{kharifcryptokitties} or shares in a venture~\cite{fenu2018ico}. The latter fueled a 12 billion USD token-based capital-investment craze called the ``ICO boom".  DEXes are a popular type of smart contract that allow users to trade such tokens in a non-custodial manner~\cite{warren20170x}.

\subsection{Gas and Fees in Ethereum}
\label{sec:ethgasbackground}

An Ethereum transaction either sends money to a non-executing {\em account address} or sends  input data (and possibly money) to a {\em smart contract address}, representing a program stored in the network's state. Transactions are gossiped to all of the nodes in the underlying peer-to-peer network to signal their availability for inclusion in a future mined block. Transactions can be in one of three states: {\em unconfirmed} and not yet mined, {\em confirmed} and considered to have been executed, or {\em rejected} as invalid by the network of Ethereum peers.

Ethereum transactions consume \emph{gas}, a pseudo-currency reflecting the number of computational steps performed by a miner (and other network nodes) executing a transaction. Ethereum contracts may contain complex logic, loops, etc., so their gas consumption can only be determined via execution. Ethereum meters gas consumption via a fixed mapping from contract op-codes to units of gas~\cite{wood2017ethereum}.

Every transaction submitted to the network for mining specifies a {\em gas price}, the per-gas-unit rate the sender will pay in Ether (ETH). The gas price times the units of gas consumed determines the {\em fee} in ETH paid by the transaction sender to the node that ultimately mines the transaction. (Bitcoin fees, by contrast, depend simply on transaction byte lengths.)

Transactions must also specify a \emph{gas limit}, the maximum number of steps a network node should attempt before rejecting a transaction. The gas limit prevents infinite loops and other DoS vectors, and allows immediate verification that a sender has adequate available funds to pay the transaction fee---up to gasPrice $\cdot$ gasLimit ETH. 

Clients can also perform what is called a ``gas replacement" transaction, resubmitting a transaction with a higher transaction fee (gasPrice $\cdot$ gasLimit) in the hope that miners will more quickly incorporate the transaction into a block.  The mechanism for doing this is nonce-based.  Each transaction issued on the network carries a \emph{nonce}, and valid canonical blockchains must only include one transaction per (account, nonce) pair on the network. When a user reissues an unconfirmed transaction with the same (account, nonce) pair but a higher gas price, a miner will prefer the reissued transaction---with its ostensibly higher transaction fee---to the replaced one.

Note that a higher gas price only corresponds to a higher fee if the reissued transaction {\em uses the same amount of gas as the replaced one}, which may not hold true. Mining software assumes this to be the case, in part because computing gas consumption incurs the computational burden of executing a contract. Arbitrageurs leverage this heuristic to pay reduced transaction fees, as we detail in Section~\ref{sec:modelproperties}.

\subsection{Continuous-Limit Exchanges}

Classic exchanges for trading stocks, commodities, and even cryptocurrencies generally share an accepted and common exchange design known as a \emph{continuous-limit orderbook}. Such an orderbook consists of a list of all open offers from buyers and sellers in the system.  Prospective buyers place a \emph{limit buy} order, which specifies a maximum price at which they are willing to buy an asset; sellers correspondingly place \emph{limit sell} orders.  A centralized counterparty, the \emph{exchange operator}, matches buyers and sellers, completing transactions automatically when there is a sell order on the books with a lower price than a buy order on the books.  Orders are matched and / or placed on the books continuously by the exchange operator, which processes orders as quickly as possible in the order they are received.  As soon as orders match, they are processed by the operator and trigger balance changes.

\subsection{Decentralized Exchanges (DEXes)}
DEXes manage continuous-limit order books using smart contracts. Traders / users hold their assets on chain and the smart contract plays the traditional role of exchange operator.

Order books are typically maintained off-chain. In some DEX designs, a counterparty selects a fresh order in the order book and presents it to the smart contract with a signed counterorder. The smart contract executes the order and counterorder, clearing the order from the order book. Traders themselves thus perform order matching. This approach is used by Etherdelta and some applications of the 0x protocol. In an alternative approach, used by, e.g., IDex and Paradex, the exchange itself performs matching off chain and submits order / counterorder pairs to the smart contract for processing. 

A more radical DEX design, called an \emph{automated market maker}~\cite{berg2009hanson}, bypasses order books altogether. The DEX consists of a smart contract that itself holds a reserve of tokens and/or Ether. Consider two assets $A$ and $B$. The contract allows a user to trade between $A$ and $B$ at any time using its reserves as a counterparty, at a set rate.  If a user buys $A$ using $B$, the price of $A$ denominated in $B$ offered by the smart contract for the \emph{next trade} is increased.  If a user instead sells $A$ for $B$, the price decreases. In this way, single parties can trade without counterparty discovery or matching.  Consequently, constant arbitrage between these and other exchanges is required to keep the rate offered in lockstep with the market rate for a commodity.  Uniswap~\cite{Uniswap:2019} and Bancor~\cite{Bancor:2019} are examples of such exchanges.

\subsection{Frontrunning and Profits through Arbitrage}

Traditional exchanges experience a classic form of predatory market behavior called frontrunning~\cite{bernhardt2008front}. In regulated markets, frontrunning is often illegal, and has resulted in prosecutions~\cite{manahov2016front} and tarnished the reputations of financial institutions caught practicing it.

Frontrunning generally exploits information asymmetries created by power structures within a financial structure, e.g., brokers having privileged access to user information.  Because there is no single party playing the role of a broker in decentralized systems, information asymmetries can arise for actors in advantageous positions in underlying infrastructure.

Frontrunning can also occur based on changes in public market information (for example, reacting to breaking news that impacts stock prices). In this form, it is not only legal, but serves as the basis of a multi-billion-dollar high-frequency trading economy.  One potential source of profit is price discrepancies across exchanges trading the same or correlated assets.  Another is information asymmetries in the speed of processing or interpreting news.  

Automated market bots from high frequency trading firms compete to profit from both at extremely high speed. They regularly build physical networks costing billions of dollars and approaching speed-of-light transmission across considerable distances.  Many economists view this behavior as a zero-sum game that profits exchanges in the long term, and argue formally that the existence of such rents is a fundamental limitation of market design that is a natural consequence of arbitrage opportunities across exchanges~\cite{budish2015high}.  This, however, remains a controversial viewpoint.

We explore both cases where bots frontrun user orders and cancellations directly, e.g., in the event of a typo or market structure weakness, and cases where bots exploit market inefficiencies to extract rent. We argue that both degrade the economic security of the underlying consensus protocol.

\section{From Decentralization to Arbitrage}

In this section, we take a deep dive into a particular example of frontrunning, arbitrage, and high-frequency automated trading on a decentralized exchange.  This concrete example will provide context for the remainder of our discussions on modeling this market and of its impact on the security of the underlying smart contract systems.

One source of potential profit, price differences, seems inherent in an environment such as smart-contract-based exchange.  Today, blockchains operate with transactions processed in discrete batches (blocks).  Furthermore, transactions are inherently dependent and therefore serial: order failures depend on past order attempts, and in some exchanges prices depend directly on order history. With multiple exchanges operating on the same system, it is possible that price differences will occur across exchanges while transactions in a block, and therefore trades on exchanges, are executed sequentially.

\subsection{Smart-Contract-Enabled Trade Atomicity}

Smart contract arbitrage opportunities have an additional, distinctive characteristic absent in traditional cross-exchange arbitrage.  Because of the atomic batch-based processing of transactions, and because transactions can themselves be initiated by smart contracts, is possible to build bots that trade across exchanges through \emph{proxy contracts}.  These proxy contracts can execute \emph{batches} of orders sequentially within a single transaction, reverting previous trades by throwing an exception if any trade in the batch fails.

This means arbitrageurs have the opportunity to compose single transactions that execute multiple trades across multiple exchanges atomically, with an all-or-nothing failure model.  One example of such a transaction is buying an asset for price $x$ and selling it immediately for price $x' > x$; if performed atomically, these transactions together generate guaranteed revenue in the base asset. For example, a smart contract proxy could execute a trade buying a type-$X$ token for $2$ ETH, and another selling it for $3$ ETH.  If both orders are on the books on some decentralized exchange, a smart contract executing both guarantees a revenue to the arbitrageur of $1$ ETH.

In traditional cross-exchange arbitrage, trades are viewed probabilistically, as there is a high likelihood that one of two trades will succeed while the other fails.  This makes smart-contract-based arbitrage in many ways simpler to observe, analyze, and study than traditional cross-exchange arbitrage, as bot intent is often explicit in order requests.

In our measurements, we focus on a small subset of these multi-trade arbitrage opportunities, which potentially involve multiple decentralized exchanges.  We restrict our focus to \emph{pure revenue opportunities}. In these opportunities, a single smart-contract based transaction executes multiple trades across one or more exchanges, and the transaction generates revenue for the trader in every traded asset.  A range of more complex and nondeterministic bot behaviors exist, described in Appendix~\ref{sec:complexnondeterminism}, but they are outside the scope of this work.

\subsection{Pure Revenue by Example}

\begin{figure}[h!]
    \centering
    \scalebox{0.8}{% Graphic for TeX using PGF
% Title: /home/tyler/revenueopportunity.dia
% Creator: Dia v0.97.3
% CreationDate: Fri Mar 29 13:04:08 2019
% For: tyler
% \usepackage{tikz}
% The following commands are not supported in PSTricks at present
% We define them conditionally, so when they are implemented,
% this pgf file will use them.
\Large{
\ifx\du\undefined
  \newlength{\du}
\fi
\setlength{\du}{15\unitlength}
\begin{tikzpicture}[scale=0.19, every node/.style={scale=0.6}]
\pgftransformxscale{1.000000}
\pgftransformyscale{-1.000000}
\definecolor{dialinecolor}{rgb}{0.000000, 0.000000, 0.000000}
\pgfsetstrokecolor{dialinecolor}
\definecolor{dialinecolor}{rgb}{1.000000, 1.000000, 1.000000}
\pgfsetfillcolor{dialinecolor}
\definecolor{dialinecolor}{rgb}{1.000000, 1.000000, 1.000000}
\pgfsetfillcolor{dialinecolor}
\pgfpathellipse{\pgfpoint{86.428364\du}{57.526682\du}}{\pgfpoint{4.428364\du}{0\du}}{\pgfpoint{0\du}{3.526682\du}}
\pgfusepath{fill}
\pgfsetlinewidth{0.050000\du}
\pgfsetdash{}{0pt}
\pgfsetdash{}{0pt}
\pgfsetmiterjoin
\definecolor{dialinecolor}{rgb}{0.000000, 0.000000, 0.000000}
\pgfsetstrokecolor{dialinecolor}
\pgfpathellipse{\pgfpoint{86.428364\du}{57.526682\du}}{\pgfpoint{4.428364\du}{0\du}}{\pgfpoint{0\du}{3.526682\du}}
\pgfusepath{stroke}
% setfont left to latex
\definecolor{dialinecolor}{rgb}{0.000000, 0.000000, 0.000000}
\pgfsetstrokecolor{dialinecolor}
\node at (86.428364\du,57.871126\du){ETH};
\definecolor{dialinecolor}{rgb}{1.000000, 1.000000, 1.000000}
\pgfsetfillcolor{dialinecolor}
\pgfpathellipse{\pgfpoint{57.800000\du}{42.200000\du}}{\pgfpoint{18.800000\du}{0\du}}{\pgfpoint{0\du}{5.200000\du}}
\pgfusepath{fill}
\pgfsetlinewidth{0.050000\du}
\pgfsetdash{}{0pt}
\pgfsetdash{}{0pt}
\pgfsetmiterjoin
\definecolor{dialinecolor}{rgb}{0.000000, 0.000000, 0.000000}
\pgfsetstrokecolor{dialinecolor}
\pgfpathellipse{\pgfpoint{57.800000\du}{42.200000\du}}{\pgfpoint{18.800000\du}{0\du}}{\pgfpoint{0\du}{5.200000\du}}
\pgfusepath{stroke}
% setfont left to latex
\definecolor{dialinecolor}{rgb}{0.000000, 0.000000, 0.000000}
\pgfsetstrokecolor{dialinecolor}
\node at (57.800000\du,42.372222\du){};
% setfont left to latex
\definecolor{dialinecolor}{rgb}{0.000000, 0.000000, 0.000000}
\pgfsetstrokecolor{dialinecolor}
\node[anchor=west] at (57.800000\du,42.200000\du){};
\definecolor{dialinecolor}{rgb}{1.000000, 1.000000, 1.000000}
\pgfsetfillcolor{dialinecolor}
\pgfpathellipse{\pgfpoint{86.950000\du}{23.466847\du}}{\pgfpoint{31.000000\du}{0\du}}{\pgfpoint{0\du}{9.416847\du}}
\pgfusepath{fill}
\pgfsetlinewidth{0.050000\du}
\pgfsetdash{}{0pt}
\pgfsetdash{}{0pt}
\pgfsetmiterjoin
\definecolor{dialinecolor}{rgb}{0.000000, 0.000000, 0.000000}
\pgfsetstrokecolor{dialinecolor}
\pgfpathellipse{\pgfpoint{86.950000\du}{23.466847\du}}{\pgfpoint{31.000000\du}{0\du}}{\pgfpoint{0\du}{9.416847\du}}
\pgfusepath{stroke}
% setfont left to latex
\definecolor{dialinecolor}{rgb}{0.000000, 0.000000, 0.000000}
\pgfsetstrokecolor{dialinecolor}
\node at (86.950000\du,23.639069\du){};
% setfont left to latex
\definecolor{dialinecolor}{rgb}{0.000000, 0.000000, 0.000000}
\pgfsetstrokecolor{dialinecolor}
\node[anchor=west] at (86.950000\du,23.466847\du){};
% setfont left to latex
\definecolor{dialinecolor}{rgb}{0.000000, 0.000000, 0.000000}
\pgfsetstrokecolor{dialinecolor}
\node[anchor=west] at (86.950000\du,23.466847\du){};
\definecolor{dialinecolor}{rgb}{1.000000, 1.000000, 1.000000}
\pgfsetfillcolor{dialinecolor}
\pgfpathellipse{\pgfpoint{117.806728\du}{42.201682\du}}{\pgfpoint{18.806728\du}{0\du}}{\pgfpoint{0\du}{5.201682\du}}
\pgfusepath{fill}
\pgfsetlinewidth{0.050000\du}
\pgfsetdash{}{0pt}
\pgfsetdash{}{0pt}
\pgfsetmiterjoin
\definecolor{dialinecolor}{rgb}{0.000000, 0.000000, 0.000000}
\pgfsetstrokecolor{dialinecolor}
\pgfpathellipse{\pgfpoint{117.806728\du}{42.201682\du}}{\pgfpoint{18.806728\du}{0\du}}{\pgfpoint{0\du}{5.201682\du}}
\pgfusepath{stroke}
% setfont left to latex
\definecolor{dialinecolor}{rgb}{0.000000, 0.000000, 0.000000}
\pgfsetstrokecolor{dialinecolor}
\node at (117.806728\du,42.373904\du){};
% setfont left to latex
\definecolor{dialinecolor}{rgb}{0.000000, 0.000000, 0.000000}
\pgfsetstrokecolor{dialinecolor}
\node[anchor=west] at (117.806728\du,42.201682\du){};
\pgfsetlinewidth{0.100000\du}
\pgfsetdash{}{0pt}
\pgfsetdash{}{0pt}
\pgfsetbuttcap
{
\definecolor{dialinecolor}{rgb}{0.000000, 0.047059, 1.000000}
\pgfsetfillcolor{dialinecolor}
% was here!!!
\pgfsetarrowsend{latex}
\definecolor{dialinecolor}{rgb}{0.000000, 0.047059, 1.000000}
\pgfsetstrokecolor{dialinecolor}
\draw (65.253408\du,37.410091\du)--(73.682693\du,31.993038\du);
}
\pgfsetlinewidth{0.100000\du}
\pgfsetdash{}{0pt}
\pgfsetdash{}{0pt}
\pgfsetbuttcap
{
\definecolor{dialinecolor}{rgb}{0.000000, 0.501961, 0.000000}
\pgfsetfillcolor{dialinecolor}
% was here!!!
\pgfsetarrowsend{latex}
\definecolor{dialinecolor}{rgb}{0.000000, 0.501961, 0.000000}
\pgfsetstrokecolor{dialinecolor}
\draw (100.842025\du,31.901467\du)--(109.988726\du,37.454939\du);
}
\pgfsetlinewidth{0.100000\du}
\pgfsetdash{}{0pt}
\pgfsetdash{}{0pt}
\pgfsetbuttcap
{
\definecolor{dialinecolor}{rgb}{0.000000, 0.501961, 0.000000}
\pgfsetfillcolor{dialinecolor}
% was here!!!
\pgfsetarrowsend{latex}
\definecolor{dialinecolor}{rgb}{0.000000, 0.501961, 0.000000}
\pgfsetstrokecolor{dialinecolor}
\draw (108.517131\du,46.738664\du)--(90.225215\du,55.672323\du);
}
% setfont left to latex
\definecolor{dialinecolor}{rgb}{0.000000, 0.501961, 0.000000}
\pgfsetstrokecolor{dialinecolor}
\node[anchor=west] at (67.500000\du,48.700000\du){};
% setfont left to latex
\definecolor{dialinecolor}{rgb}{0.000000, 0.000000, 0.000000}
\pgfsetstrokecolor{dialinecolor}
\node[anchor=west] at (65.950000\du,56.050000\du){0.142123};
% setfont left to latex
\definecolor{dialinecolor}{rgb}{0.000000, 0.000000, 0.000000}
\pgfsetstrokecolor{dialinecolor}
\node[anchor=west] at (68.450000\du,48.950000\du){};
% setfont left to latex
\definecolor{dialinecolor}{rgb}{0.000000, 0.000000, 0.000000}
\pgfsetstrokecolor{dialinecolor}
\node[anchor=west] at (67.800000\du,49.050000\du){};
% setfont left to latex
\definecolor{dialinecolor}{rgb}{0.000000, 0.000000, 0.000000}
\pgfsetstrokecolor{dialinecolor}
\node[anchor=west] at (68.850000\du,48.600000\du){};
% setfont left to latex
\definecolor{dialinecolor}{rgb}{0.000000, 0.000000, 0.000000}
\pgfsetstrokecolor{dialinecolor}
\node[anchor=west] at (69.350000\du,48.750000\du){};
% setfont left to latex
\definecolor{dialinecolor}{rgb}{0.000000, 0.000000, 0.000000}
\pgfsetstrokecolor{dialinecolor}
\node[anchor=west] at (50.000000\du,35.000000\du){1.55496e+08};
% setfont left to latex
\definecolor{dialinecolor}{rgb}{0.000000, 0.000000, 0.000000}
\pgfsetstrokecolor{dialinecolor}
\node[anchor=west] at (109.000000\du,35.000000\du){1.55e+08};
% setfont left to latex
\definecolor{dialinecolor}{rgb}{0.000000, 0.000000, 0.000000}
\pgfsetstrokecolor{dialinecolor}
\node[anchor=west] at (95.000000\du,56.000000\du){0.93};
% setfont left to latex
\definecolor{dialinecolor}{rgb}{0.000000, 0.000000, 0.000000}
\pgfsetstrokecolor{dialinecolor}
\node[anchor=west] at (78.000000\du,69.000000\du){};
% setfont left to latex
\definecolor{dialinecolor}{rgb}{0.000000, 0.000000, 0.000000}
\pgfsetstrokecolor{dialinecolor}
\node[anchor=west] at (77.000000\du,69.000000\du){};
% setfont left to latex
\definecolor{dialinecolor}{rgb}{0.000000, 0.000000, 0.000000}
\pgfsetstrokecolor{dialinecolor}
\node[anchor=west] at (42.000000\du,41.000000\du){Trade \#1 (Tokenstore)};
% setfont left to latex
\definecolor{dialinecolor}{rgb}{0.000000, 0.000000, 0.000000}
\pgfsetstrokecolor{dialinecolor}
\node[anchor=west] at (42.000000\du,43.700000\du){1.09409e+09 FREE/ETH};
% setfont left to latex
\definecolor{dialinecolor}{rgb}{0.000000, 0.000000, 0.000000}
\pgfsetstrokecolor{dialinecolor}
\node[anchor=west] at (57.800000\du,42.200000\du){};
% setfont left to latex
\definecolor{dialinecolor}{rgb}{0.000000, 0.000000, 0.000000}
\pgfsetstrokecolor{dialinecolor}
\node at (86.950000\du,20.050000\du){Free Coin (FREE)};
% setfont left to latex
\definecolor{dialinecolor}{rgb}{0.000000, 0.000000, 0.000000}
\pgfsetstrokecolor{dialinecolor}
\node at (86.950000\du,24.050000\du){0x2f141ce366a2462f02cea3d12cf93e4dca49e4fd};
% setfont left to latex
\definecolor{dialinecolor}{rgb}{0.000000, 0.000000, 0.000000}
\pgfsetstrokecolor{dialinecolor}
\node[anchor=west] at (80.950000\du,37.050000\du){};
% setfont left to latex
\definecolor{dialinecolor}{rgb}{0.000000, 0.000000, 0.000000}
\pgfsetstrokecolor{dialinecolor}
\node[anchor=west] at (86.950000\du,23.466847\du){};
% setfont left to latex
\definecolor{dialinecolor}{rgb}{0.000000, 0.000000, 0.000000}
\pgfsetstrokecolor{dialinecolor}
\node[anchor=west] at (102.250000\du,41.000000\du){Trade \#2 (Tokenstore)};
% setfont left to latex
\definecolor{dialinecolor}{rgb}{0.000000, 0.000000, 0.000000}
\pgfsetstrokecolor{dialinecolor}
\node[anchor=west] at (102.200000\du,43.700000\du){1.66667e+08 FREE/ETH};
% setfont left to latex
\definecolor{dialinecolor}{rgb}{0.000000, 0.000000, 0.000000}
\pgfsetstrokecolor{dialinecolor}
\node[anchor=west] at (117.806728\du,42.201682\du){};
\pgfsetlinewidth{0.100000\du}
\pgfsetdash{}{0pt}
\pgfsetdash{}{0pt}
\pgfsetbuttcap
{
\definecolor{dialinecolor}{rgb}{0.000000, 0.000000, 1.000000}
\pgfsetfillcolor{dialinecolor}
% was here!!!
\pgfsetarrowsend{latex}
\definecolor{dialinecolor}{rgb}{0.000000, 0.000000, 1.000000}
\pgfsetstrokecolor{dialinecolor}
\draw (82.734494\du,55.549106\du)--(66.451501\du,46.831728\du);
}
\end{tikzpicture}
}}
    \vspace{-8mm}
    \caption{Example pure revenue opportunity observed in Ethereum transaction 0xc889bd13594f75e4dd824f04f0c2ad03896cb7ec6518df02455e9560367bb9c4, exploiting an orderbook cross on the TokenStore DEX.  Edges of the same color are generated by a single trade. As a pure revenue opportunity, it generates net profit in both FREE and ETH.}
    \label{fig:purerevenue}
\end{figure}
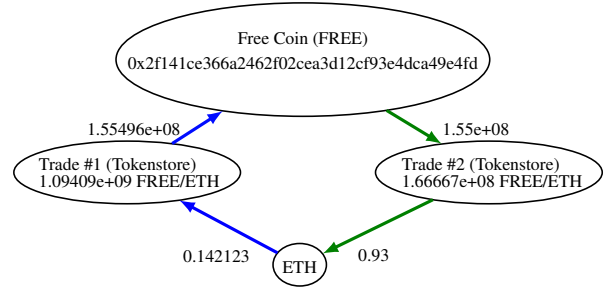

\begin{figure*}[ht!]
    \centering
    \includegraphics[width=\textwidth]{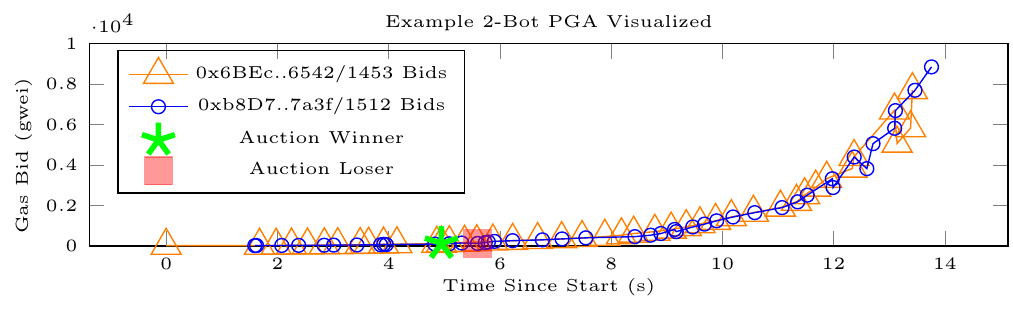} \\
    \scalebox{0.7}{
        \hspace{-6mm}
        \begin{tabular}{c|c|l|c|l}
\specialcell{\textbf{Seconds}\\\textbf{Elapsed}} & \textbf{Quantity @ Price Bid} & \hfil\textbf{Ethereum Transaction Origin (Public Key Hash)} & \textbf{Nonce} & \hfil\textbf{Transaction Hash} \vspace{1mm}\\
0.000 & 192085 @ 25.10 & \color{orange} 0x6BEcAb24Ed88Ec13D0A18f20e7dC5E4d5b146542 \color{black} & 1453 & 0xd32653ca9694a6d1299335f3c04f74cc159bee48c1d32d3a421db08c638ffc78 \\
1.593 & 231520 @ 25.00 & \color{blue} 0xb8D76f4BC2518F8eb508bf0Ccca76f8F9DD57a3f \color{black} & 1512 & 0xb901e6dc2c229fd9105448fcc23eaebdedb476c21b6c6e7ddf8d2df4e838d2c7 \\
1.624 & 231520 @ 28.75 & \color{blue} 0xb8D76f4BC2518F8eb508bf0Ccca76f8F9DD57a3f \color{black} & 1512 & 0x9f592504eb71a7452b7a395a7f5ecd34eaa5d090da1162e74221562af54c8f67 \\
1.679 & 227534 @ 28.81 & \color{orange} 0x6BEcAb24Ed88Ec13D0A18f20e7dC5E4d5b146542 \color{black} & 1453 & 0x83e2a6774654a9540c3fad8837afcc88b4c932ab2374819254f887305c3a4b22 \\

... & ... & ... & ... & ... \\

4.949 & 227534 @ 134.02 & \color{orange} 0x6BEcAb24Ed88Ec13D0A18f20e7dC5E4d5b146542 \color{black} & 1453 & \color{money} 0xc889bd13594f75e4dd824f04f0c2ad03896cb7ec6518df02455e9560367bb9c4 \color{black} \\

5.599 & 231520 @ 133.76 &  \color{blue} 0xb8D76f4BC2518F8eb508bf0Ccca76f8F9DD57a3f \color{black} & 1512 & \color{red} 0xaa86d782328c0c9c422e3f2a3170ff41ae21a27ad395c48db76b0080898f85db \color{black} \\

... & ... & ... & ... & ... \\

13.383 & 227534 @ 5834.77 & \color{orange} 0x6BEcAb24Ed88Ec13D0A18f20e7dC5E4d5b146542 \color{black} & 1453 & 0xb0dc97140394c5f65332ebc459d5e66f89099dbb4d335c866b32280270102858 \\
13.416 & 227534 @ 7716.48 & \color{orange} 0x6BEcAb24Ed88Ec13D0A18f20e7dC5E4d5b146542 \color{black} & 1453 & 0x1825be6951577e72a1dafc8de564ce1ccfe5d284173e11e77b2e7f6b1b44571c \\
13.462 & 231520 @ 7701.08 & \color{blue} 0xb8D76f4BC2518F8eb508bf0Ccca76f8F9DD57a3f \color{black} & 1512 & 0xa9823358c99149f0e6343c604c35988468d01d02868437d8251b3cee282dc92b \\
m13.759 & 231520 @ 8856.24 & \color{blue} 0xb8D76f4BC2518F8eb508bf0Ccca76f8F9DD57a3f \color{black} & 1512 & 0x366c30a534b5f3d8a6d251f97d401997624d1fe8d3af07ede4d19105dc970942 \\
    \end{tabular}}
    \caption{One example PGA that was observed over the Ethereum peer-to-peer network, resulting from the profit opportunity in Figure~\ref{fig:purerevenue}.  The top graph shows the gas bids of two observed bots over time, while the bottom table details the first and last two bids placed by each bot and the two mined bids (center).}
    \label{fig:exampleauction}
\end{figure*}

Figure~\ref{fig:purerevenue} shows one example of a pure revenue transaction, executed on November 15, 2018.   In this transaction, two trades are executed on a decentralized exchange, TokenStore, which features a design conceptually similar to that of Etherdelta.  The first executed trade buys Free Coin, an obscure token.\footnote{As of writing on Mar 13 2019, Free Coin is listed at currency rank 303 by market cap https://coinmarketcap.com/currencies/free-coin/.} By inspection, this difference in rates is a clear result of someone using the exchange API and committing an off-by-one error, offering to buy tokens at 10x the market rate.  This created a cross in the order book (sell order at more than a buy order would pay), which when both executed by the same arbitrage counterparty, generates the flow of funds in Figure~\ref{fig:purerevenue}.  While this opportunity probably arose from a typo, a variety of revenue sources exist.  For example, inconsistent price feeds, or variance across exchange designs that respond to market movements at different speeds can also create pure revenue.

Note that both orders are executed inside a single Ethereum contract, and are executed in an atomic batch through a smart contract proxy.  In general, we model each opportunity as a graph as shown in Figure~\ref{fig:purerevenue}.  Edges represent currency flows, and nodes are either exchanges nodes (which receive one currency and output another at the stated exchange rate), or asset nodes (which are either sources or sinks for same-color trade subgraphs, depending on whether they are bought or sold).  We label exchange nodes with the exchange name and trade rate, and asset nodes with the asset symbol.

The net revenue in each traded asset is the sum of inflows minus outflows in the corresponding asset node.  In the trade in Figure~\ref{fig:purerevenue}, the revenue for ETH was therefore $0.93 \text{ ETH}-0.14123 \text{ ETH} \approx .79 \text{ ETH}$, or approximately 267 USD in equivalent value (at November 15, 2018 prices of 338.15 USD per ETH).  Approximately 496,000 residual FREE was also left in the arbitrageur's account, although less significant and liquid than the ETH revenue.  

To calculate the transaction profit, we subtract the cost from the revenue, in this case the gas paid by the transaction that was mined for the arbitrageur.  Transaction 0xc889...b9c4 paid a gas price of 134.02 Gwei (the canonical unit for representing gas rates; $1$ Gwei = $10^{-9} \text{ ETH}$).  The transaction used $113,265$ gas, approximately equivalent to $113,265$ computational steps.

It is worth noting that transactions that cancel their bids can use less gas, as they pay only for attempted execution (execution requires costly computational steps).  This will be relevant to later models in Section~\ref{sec:modelproperties}.  Furthermore, complex conditional transactions, price queries, or other computationally intensive order preferences will obviously cost more, making optimizing for gas consumption important in the development of competitive bots.  The total cost of this transaction was therefore $113265 \cdot 134.02 \text{ Gwei} = 0.01518$ \text{ETH}, or around $5.13$ USD at the time of the transaction.  The associated profit was therefore $\approx .79-.01518=.77 \text{ ETH}$, or 267 USD.

\iffalse
\subsection{Web Dashboard for Arbitrage Monitoring}

We provide a dashboard for displaying all such trade revenue graphs on all supported decentralized exchanges for \emph{all} observed arbitrage transactions at \url{https://frontrun.me/revenue}; these revenue graphs and all associated transaction information, including profit totals for all assets traded, are automatically generated and parsed from Google's Bigquery for Ethereum~\cite{bqeth} once daily.  We parse exchange logs associated with trades on seven supported decentralized exchange platforms (see: Figure~\ref{fig:supportedexchanges}), which include [] of the top exchanges by volume as of this writing.  As of this writing on [], our dashboard includes such profit graphs for $x$ arbitrage transactions, of which $y$ are \emph{pure revenue} transactions.  We encourage the interested reader to explore this interactive data for greater insight into the fixed pure profit opportunity arbitrage market, and present aggregate statistics on this market in Section~\ref{sec:stats}.
\fi

\subsection{PGAs, Ordering Fees, and... HFT?}

Ethereum transactions are routed in a peer-to-peer gossip multicast protocol by client node software.  This means that all transaction information is available to all participants in this network, but earlier to participants with advantageous positions in the gossip topology. Additionally, nodes can simulate the outcome of every transaction given the current or expected system state.  Once an arbitrage transaction is submitted, therefore, the sequence of trades it involves is publicly known by the network's peer-to-peer nodes.

A natural question then arises: how is priority determined between arbitrageurs?  Because each pure profit opportunity carries some computable profit $p$ and is broadcast globally, a competitive game naturally ensues among  arbitrage bots to be the first to execute an atomic transaction that exploits the opportunity.  The mechanics of the system dictate that all subsequent transactions in the game will fail.  How this particular game plays out depends on the peer-to-peer relay network mechanics of the underlying blockchain, as well as mining pool strategies and order book designs in the underlying exchanges.  We present a greatly simplified model of this game formally in Section~\ref{sec:modelproperties}.

In Ethereum, the PGA game we observe consists of transactions issued with the same (account, nonce) pair, expressing a bid to a miner, where the miner is paid gas fees as described in Section~\ref{sec:ethgasbackground}.  We call this interaction of issuing repeated bids a \emph{priority gas auction}, or PGA.

To place a bid on an arbitrage opportunity, an arbitrageur simply issues a smart contract transaction with associated gas price $g$ that atomically bundles multiple trades (performing the arbitrage). If they later wish to increase this bid, either for strategic reasons or because they notice a higher bid that would eliminate their profit issued on the peer to peer relay network, they simply reissue the transaction with the same nonce and a higher gas price $g'$.  If two arbitrageurs are bidding against each other, what emerges is essentially an auction, as shown in Figure~\ref{fig:exampleauction}.  It is in a miner's interest to order first whichever transaction offers the highest price (if the miner does not plan on itself arbitraging the market). This approach will encourage arbitrageurs to bid each other up for transaction priority.

Figure~\ref{fig:exampleauction} shows the action that occurs in such an auction from the point of view of the network.  Each row in the associated table is a transaction observed by our Ethereum monitor on the peer to peer network, using the experimental harness described in Section~\ref{sec:prevalence}.  In this example, we see two accounts, {0x6BEc...6542} with nonce 1453 and {0xb8D7...7a3f} with nonce 1512, bidding against each other for priority.  These accounts issue transactions with ever increasing gas prices: bot 0x6BEc...6542 issues 42 transactions in 13.4 seconds, and 0xb8D7...7a3f issues 43 transactions in 12.1 seconds.  In 4.94 seconds, the auction is over.  The transaction that eventually ends up mined with priority is shown with a green star, and is transaction 0xc889...b9c4 issued by bot 0x6BEc...6542.  This bot is considered the \emph{winner}, and pays the full gas price to reap the revenue as described in the previous section.  The transaction was mined in block 6709727 by the mining pool ``MiningPoolHub".   The transaction shown in the red square, with hash 0xaa86...85db by bot 0xb8D7...7a3f, is also mined and included in the final block, as each (account, nonce) pair can include one transaction in the next block to be mined as per the protocol.  Miners are incentivized to include even failed transactions, as these transactions pay for attempted execution.  Note that while the winning transaction used 113,265 gas in execution, the losing transaction paid for 33,547 gas units, a far smaller sum.  In game theoretic terms, each auction represents a variant of \emph{all pay} auction, where instead of paying their full bid, the loser is forced to pay a percentage to the miner.
\section{Arbitrage Prevalence Measurement Study}
\label{sec:prevalence}

We now present empirical measurements of the prevalence of arbitrage. We discuss our methodology (and its limitations).

\subsection{Experimental Setup}

\begin{figure}[h]
    \centering
    \input{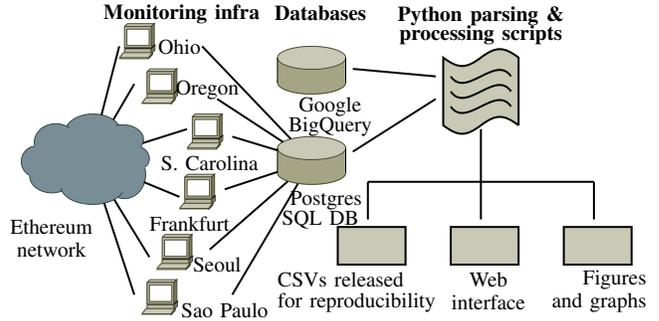}
    \caption{Deployed PGA measurement infrastructure architecture.}
    \label{fig:pgainfra}
\end{figure}

We first describe the experimental harness used to observe PGA transactions, displayed in Figure~\ref{fig:pgainfra}.  On-chain data is not sufficient to analyze PGAs, as all but the final ``winner'' transaction are discarded by the network.  Furthermore, because transactions are replaced so quickly and nodes don't propagate replaced transactions, many transactions are never even propagated to all Ethereum nodes.

No tools existed for analysis of unconfirmed and rejected transactions, so we wrote our own.  We forked the Go-Ethereum client to record unconfirmed transaction in the mempool.  We deployed six geodistributed nodes across multiple data centers, with timestamps synchronized to the nanosecond level by NTP. We collected an observation every time one of the $\approx256$ nodes peered with one of our deployed modified nodes relayed a transaction to us.  We collected nine months of data, amounting to over 300 gigabytes, including 708,385,840 unique observations of PGA arbitrage bots.  Node locations are shown in Figure~\ref{fig:pgainfra}, although not all nodes were online throughout the experiment. (We ultimately sacrificed timing resolution to reduce prohibitive ongoing costs.)

Because it is technically infeasible to store observations of every Ethereum transaction, we focus on a list of suspected arbitrage bot transactions.  This list is seeded with accounts we observe performing pure revenue transactions on the blockchain, and is updated dynamically any time a high-value gas replacement transaction is seen at an order of magnitude over current gas market price.  We also built a web interface to manage the uptime of our nodes and associated bot lists.

We supplemented this mempool data we collected on PGAs with on-chain data sourced from Google's BigQuery Ethereum service (as well as other on-chain metadata, like block timestamps), allowing us to parse logs for successful transactions and determine their profit.  We also used daily price data from \url{coinmetrics.io} for USD conversions.

We then developed a suite of Python scripts to combine and analyze this data.  These scripts use a heuristic to place all observations on a timeline, identifying a PGA whenever a high-value gas replacement transaction occurs.  All transactions in a time interval around this observation are then considered part of the ``auction," and broken down per bot.  The scripts also aggregate meta-statistics on PGAs, calculating strategy and latency trends in observed bots.  This data is used to generate both our web interface and the figures in this paper.  We release all source and processed/derived data in CSV format for further analysis by the community on our Github.

\subsubsection{Instrumentation Limitations}
\label{sec:limitations}

The above instrumentation has limitations that may affect our data quality.  We describe them at a high level here, and mention any impact these limitations may have on our derived results throughout.

It is possible that a transaction will be replaced before reaching our nodes, though this seems unlikely, 
since each of our observed transaction tends to have hundreds of associated observations. More critically, our time-slicing method for identifying PGAs may lump together unrelated bot activity into a single ``auction.''  Manual inspection suggests that because auctions are relatively infrequent (every few hours), the majority consist of correlated bot activity.  Time-slicing could also harvest unrelated transactions from other arbitrage bots that are not PGA behavior; we prune these in our aggregate statistics by only including bots that have placed at least 4 ``bids'' in an observed PGA, and as a result may lack data on bots that place $<4$ bids.  We also may miss bots whose addresses are not in our PGA lists, leading to missing bidders in certain auctions.

Lastly, our instrumentation calculates pure revenue opportunities by parsing transaction logs on supported exchanges for transactions that contain more than two trades executed by a smart contract.  We support only a limited subset of popular DEXes, omitting revenue opportunities on unsupported exchanges.  Because we aim to establish a lower bound (i.e. a potentially too conservative result), we feel this is acceptable.  Our supported exchanges include the top five DEXes by sustained volume at the time of infrastructure development.

This subset still proves sufficient to provide substantial insight on the PGA market not afforded to regular nodes in a blockchain context, highlighting the limits of the ``transparency" afforded users by these systems.

\subsection{Observations}

We now describe key results of our data gathering.

\begin{figure}[h]
\scalebox{0.9}{\includegraphics{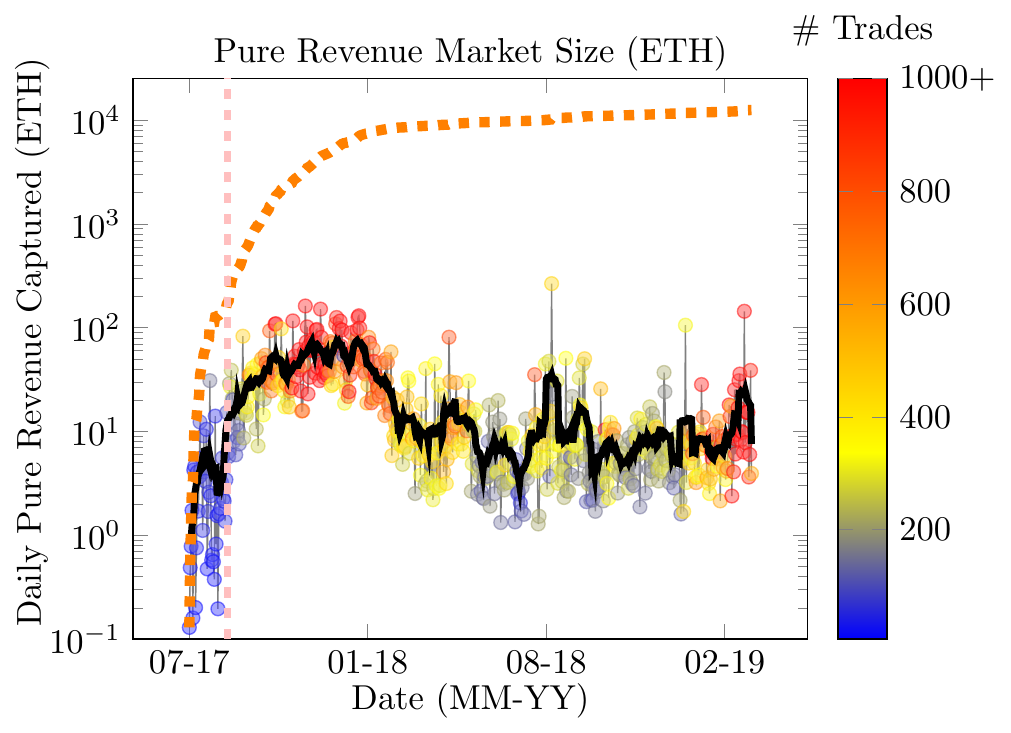}}
\vspace{-4mm}
\caption{Lower bound size of the observed pure revenue market, ETH.  Black line shows a 14-day moving average of market size, while the dotted orange line plots cumulative market size.  Scatter plot colorings/shading indicates number of daily pure revenue trades.  The pink vertical line shows the release of~\cite{blogpost} associated with early versions of this work becoming public.}\label{fig:purerevenueethscatter}
%133815 7408826 3875490
\end{figure}

\begin{figure}[h]
%\hspace{6mm}
\scalebox{1.0}{\includegraphics{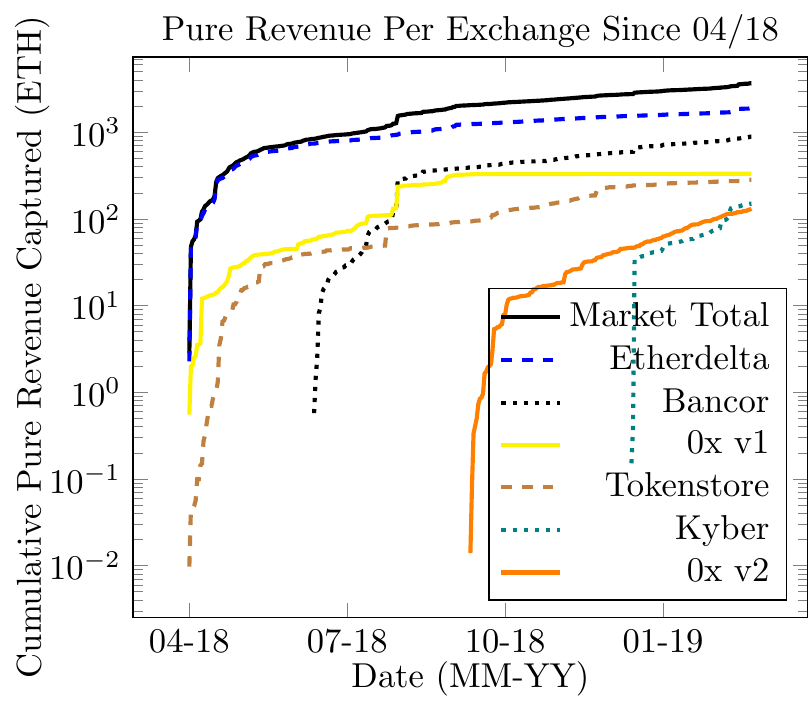}}
\vspace{-3mm}
\caption{Top exchanges by cumulative pure revenue offered bots since 04/18.}\label{fig:purerevenueexchbreakdown}
\end{figure}

\begin{figure*}[h]
\scalebox{1.0}{\includegraphics{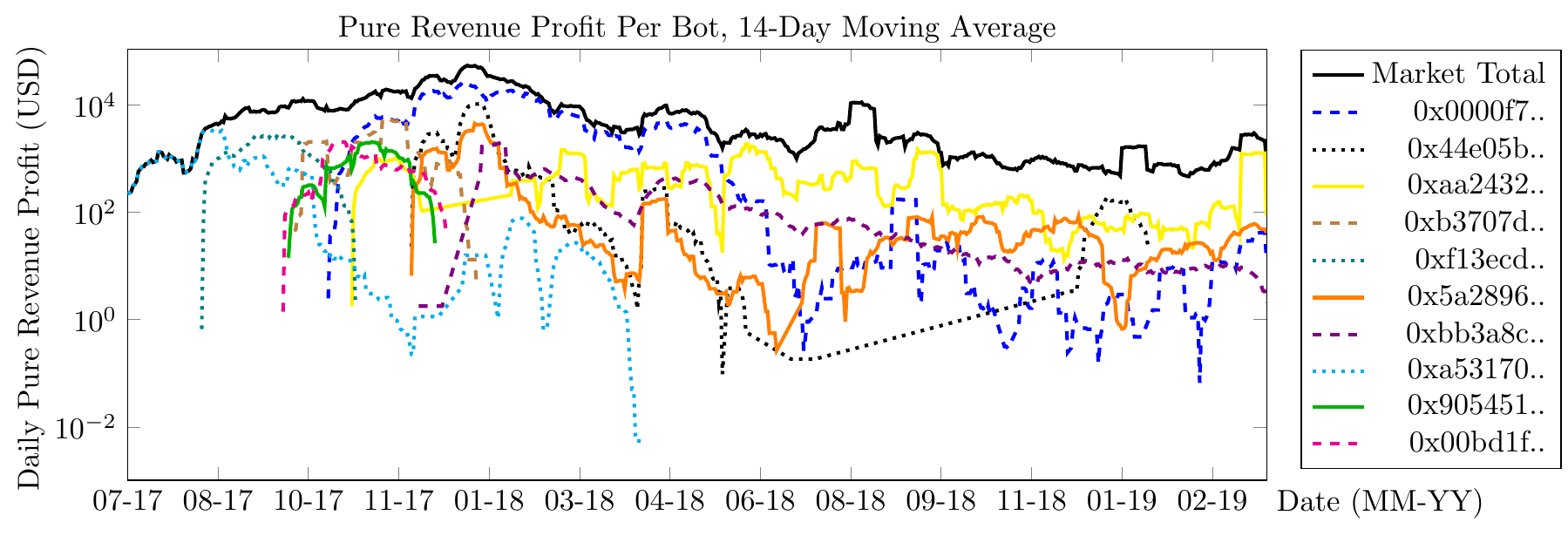}}
\vspace{-7mm}
\caption{Profit (revenue minus estimated costs) for pure revenue bots over time. We observe that the median profit is 65\% of the size of the pure revenue opportunity across 138,948 observed pure revenue transactions.}\label{fig:purerevenuebotprofit}
\vspace{-4mm}
\end{figure*}

\begin{figure}[h]
\scalebox{0.9}{\includegraphics{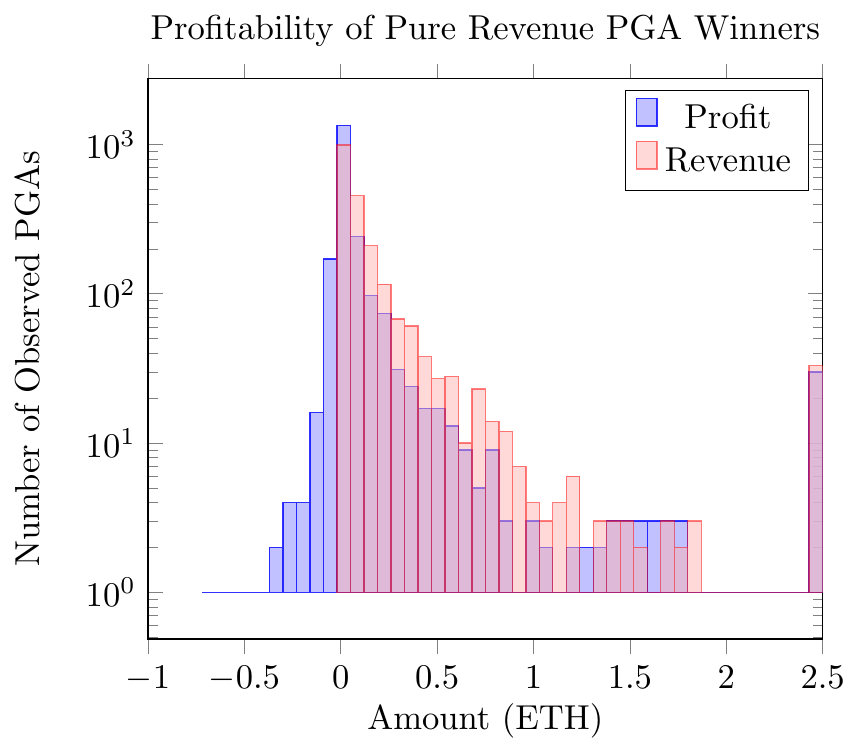}}
\vspace{-2mm}
\caption{Distribution of profit and revenue in competitive PGAs.  We focus on amounts from -1 to 2.5 ETH, which by inspection comprise the majority of the distribution body.  29 of 2,135 profits fall outside this range.}\label{fig:revenueprofithist}
\vspace{-4mm}
\end{figure}

Figure~\ref{fig:purerevenueethscatter} shows the size (breadth) of the arbitrage market, denominated in ETH.  We choose to use ETH denominations in the rest of the work as they display a relatively consistent distribution, as shown in the figure.  USD revenues are more volatile, as is visible in the market revenue graph discussed in Appendix~\ref{sec:purerevenueusd}.

This distribution is consistent with common understanding of arbitrage as sourced from market structure design, as ETH is the most liquid traded token on this market.  As seen in this graph, after the initial market development in late 2017, a relatively active period of arbitrage occurred during which bots often performed over 1,000 daily trades for a daily revenue of 10-100 ETH.  Later, the market matured into a more steady and consistent profit distribution, with 1-10 ETH available for arbitrage daily.  Recently, an increasing number of pure revenue trades is observed, showing the trend of the DEX market to smaller retail transactions and more efficient market designs, decreasing the average opportunity size but presenting more frequent opportunities to a more efficient bot market.  The limitations in our instrumentation mean, however, result in an underestimate of market size, as described in Section~\ref{sec:limitations}.

Interestingly, Figure~\ref{fig:purerevenueethscatter} also highlights the origin story of this work.  In August 2017, we published a preliminary report on the dangers of decentralized arbitrage and associated bot designs in~\cite{blogpost}.  The date of our blog release is shown as the vertical line on the graph.  As part of this post, we actively probed the Etherdelta orderbook for pure revenue opportunities, measuring available profit at 4,472.75 USD / day or 1.6M+ USD / year.  We executed our own trading bot to confirm the technical feasibility of profit, observing a total 58\% rate of success in capturing pure revenue for our own bots (soon outcompeted by other bots after the release of our post).  We calculated the expected daily profit of an arbitrage bot at the time, before most decentralized exchanges were online, at 0.32 ETH/day $\cdot 58.3\%$ success $\cdot 45 - 0.004 \cdot (1 - 58.3\%)$ failure cost * $45$ observed opportunities = 8.32 ETH (2500 USD) daily profit. Note in Figure~\ref{fig:purerevenuebotprofit} that the majority of profitable bots today joined shortly after our public release, which inadvertently sparked a thriving cottage bot economy!

Figure~\ref{fig:purerevenueexchbreakdown} shows the breakdown of the pure revenue market by exchanges, since 04/18 (the first significant pure revenue opportunities outside Etherdelta, representing new exchanges coming online).  An oligopolistic pure revenue market is observed, with Etherdelta generating the majority of observed pure revenue.  Nonetheless, a range of other exchanges furnishes a growing and relatively consistent distribution of opportunities for bots.  Because we only support certain exchanges, unsupported exchanges may be missing from this graph, and the market total provides a lower bound only.

Figure~\ref{fig:purerevenuebotprofit} shows the top ten transaction senders we observe in the competitive pure revenue market, and their associated profits.  The cost we measure for these bots includes only the gas fees for the pure revenue transactions they make, and thus may underestimate their total fees (e.g. for failed transactions).  We also do not include any bots that do not trade on our supported exchanges.  Note that we make no effort to group transaction senders together to consolidate multiple senders operated by a single entity; some heuristics are available for doing this (e.g. monitoring use of similar arbitrage contracts), but we leave such aggregation to future work.  
We focus on USD on this graph, as we believe that most actors use fiat currencies to pay participation costs.

As in the exchange breakdown, this figure suggests an oligopolistic market, with single bots often dominating the profit space for extended time periods (for example, 0xaa24... dominates the recent pure revenue space, where 0x0000... experienced a long period of dominance through late 2017 and early 2018).  This diagram also shows the behavior of top bots exiting the market after failing to update their strategies.  For example 0xa53... exits around 03/18, and 0xf13... exits in late 2017, early in the market.  Nonetheless, many bots enjoy extended runs of profitability, and continue operating for years.

Interestingly, the revenue and profit graphs for these bots have nearly identical shapes (the revenue graph is described in Appendix~\ref{sec:purerevenueusd}). Profit-to-revenue ratios are relatively well distributed, with a median of 65\% of the revenue of an opportunity captured by the winning bot according to our profit heuristics.  Most opportunities are also uncontested, providing the bots a steady stream of income.

Thusfar, we have focused exclusively on pure revenue opportunities, which may lead to PGAs or may be claimed by a single bot uncontested.  A natural question becomes whether the profitability of opporutnities persists in a competitive market, or whether competition among bots creates a generally unprofitable zero-sum or negative-sum environment.

Figure~\ref{fig:revenueprofithist} describes the breakdown between the profit and revenue of observed pure revenue opportunities involved in priority gas auctions, indicating that at least one bot placed a competitive gas replacement bid to capture such an opportunity.  While there is a spike around 0 profit, indicating that many PGA opportunities \emph{are} zero or negative sum for their players, the vast majority of these opportunities see relatively insignificant costs, and the profit distribution still provides a mean profit to winning bots that offers them a majority of associated revenue.  This confirms our suspicious that bots may be engaging in uncoordinated cooperation to maintain the profitability of the PGA market at the expense of miners, which we model fully in Section~\ref{sec:model}.

\begin{figure}[h]
\scalebox{0.9}{\includegraphics{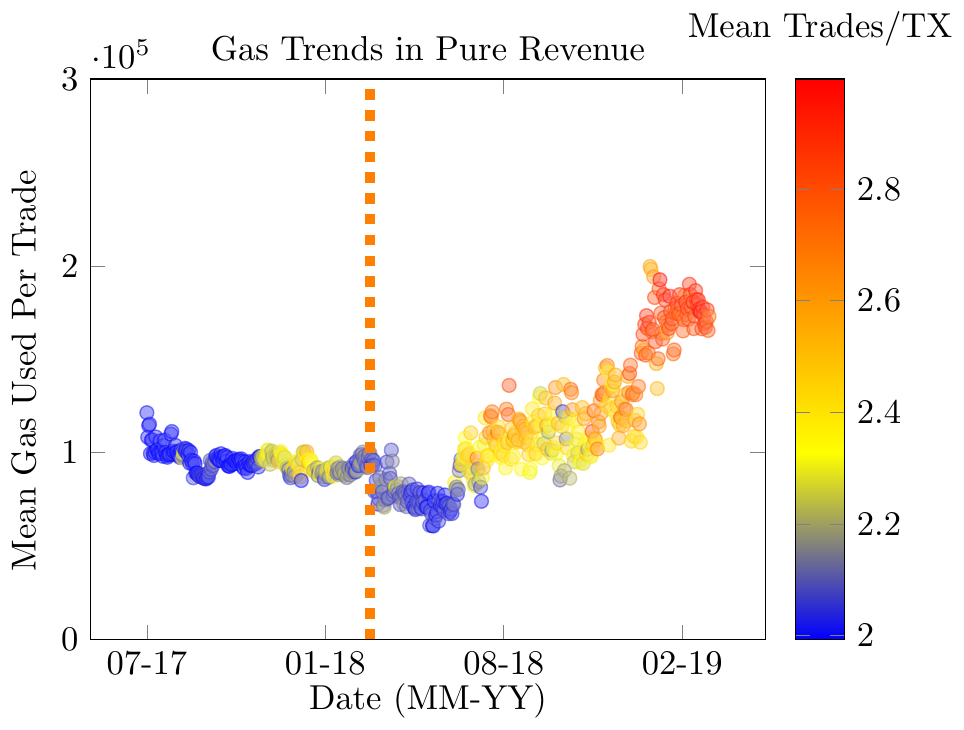}}
\vspace{-7mm}
\caption{Trends in quantity of gas used per trade by observed pure revenue transactions. The vertical line indicates the release date of a gas optimization token by the authors and others as part of this study, described in~\cite{gastoken}.}\label{fig:gastrendsscatter}
\end{figure}

A final natural question to ask about the market is the extent to which the gas used in a transaction (representing the ``quantity bid'' in a PGA) is an important competitive optimization vector for pure revenue bots.  Figure~\ref{fig:gastrendsscatter} shows this optimization over time.  Before 04/18, primarily Etherdelta trends are observed.  As we expect, we see bots optimizing their gas costs with a clear downward trend over time, despite relatively consistent execution of primarily pure revenue transactions containing two trades per transaction.

The advent of more sophisticated exchanges shows market maturation, with more complex opportunities executed that require more gas per trade.  The number of trades per transaction has also increased, now averaging well over two trades per transaction.  Part of this trend is due to ``liquidity pooling" exchange designs like Kyber, which allow use of an exchange contract as a proxy to trade on other exchanges.

The dotted orange line on this graph shows the authors' public release of a token called GasToken~\cite{gastoken}. It leverages a feature in Ethereum's incentive model enabling arbitrageurs to perform \emph{gas arbitrage} over time, banking gas at below market rates and deploying it to win PGAs. The mechanism by which this gas arbitrage is performed yields a transaction refund in the quantity bid to miners, tricking miners into accepting smaller than expected bids.  Because this allows bots to bid higher gas prices for less quantity at the same level of cost, this token is now a requirement for participating competitively in the pure revenue arbitrage market.  This is reflected in our trends graph as a sharp downtrend in gas quantity bots demand from miners after release of our token.  This anecdote provides unusual insight on consequences of performing active experiments on emerging competitive markets.

\subsection{User Comments}

One unique feature of Ethereum over other market designs is the public nature of some of its trading data.  For example, the addresses of the arbitrageurs in Figure~\ref{fig:purerevenuebotprofit} are known, and can link different trades by a single actor.  This often allows users to see which arbitrage bot targeted a transaction they may have made, and to comment on bot activity on generic public comment pages on blockchain explorers.  The top sender in Figure~\ref{fig:purerevenuebotprofit}, for example, appears to often profit off typographical errors users make in their decentralized exchange trades.\footnote{These comments were sourced from \url{https://etherscan.io/address/0x0000F7F39325076881E5fC566E99595542532aE2\#comments}, though similar pages exist for several of the bots in Figure~\ref{fig:purerevenuebotprofit} and our data.}  One user, ``Getcoin Hub Inc.", pleads with this arbitrage bot, saying ``This was obviously a mistake transaction - any chance you can find it in your heart to send it back?".  Another, Benjamin Huffman, pleads with the bot that ``I am a single parent trying to make ends meet at a job I hate. Please have some mercy."  Yet another, Alfie, leaves a comment requesting the return of their ETH, stating "Please I ask you to send me the ETH back as I really need it to continue my education. I might need to sell my car to pay what is already due for this semester."  Worse still, one user, Rajesh Kumar, claims ``sir this was error order. please give back eth sir. this all of my village rupee I trade for them. if i do not receive back eth i will be in the **** sir please".

While it is impossible to verify the authenticity of these anonymous comments, their existence does point to the effect that exchanges with bad usability properties and flawed market designs can have on the guarantees provided their users.  The existence of multiple users losing a large chunk of funds through such flawed exchange designs embodies the need for carefully considered and formalized guarantees for DEXes.
\section{Priority Gas Auction (PGA) Modeling}
\label{sec:model}

To shed light on the strategies we observe in practice, we present a formal model for PGAs in this section. We then explain the notion of {\em advantage} that we use to understand strategies' effectiveness, and study several classes of observed strategies in the context of our model. 

We provide informal, high-level intuition about the games both miners and bots play in Appendix~\ref{sec:intuition}.

\subsection{Model properties}
\label{sec:modelproperties}

We model a PGA as a sequential game among a set of $n$ players $\{\player_0, \player_1, \ldots, \player_{n-1}\}$ who bid against one another to obtain a payoff of $\payoff$.\footnote{In this ``dollar auction,'' $\payoff$ represents a normalization of any payoff amount w.l.o.g., and includes gas costs for the winner.} 

This game has the following properties, which model blockchain dynamics:

\vspace{4mm}
\paragraph{\bf Continuous time} Players act in continuous time, rather than discrete rounds (as in typical extensive-form games). This is because blockchain networks are asynchronous. 

\paragraph{\bf  Imperfect information} Players eventually see one another's bids, but not immediately, a feature modeling blockchain network {\em latency}. Player $\player_i$ observes other players' bids after some fixed time $\latency_i$. Players may have differing latencies, small ones conferring a competitive advantage. 

We measure latency as a relative figure to the latency of the miners, which are typically the best connected nodes in the network. Thus $\latency_i = 0$ indicates that $\player_i$ has a superior network position and observes other player's bids with no additional delay over that of the miners. 

\paragraph{\bf  All pay} In PGAs, losing players pay gas costs for their failed transactions. Our model captures this cost by having a losing player $\player$ pay $\lossmult(\bid_{\sf last})$, where $\bid_{\sf last}$ is the last bid made by $\player$ and $\lossmult()$ is a {\em loss function}.\footnote{As we explain above, bots use various tricks to reduce $\lossmult$.} This type of all-pay auction in which players can submit multiple bids is known as a {\em dollar} auction~\cite{shubik1971dollar}, in which not only the winner, but losers pay. In our study, we typically observe auctions that are {\em partial all-pay}, i.e., $\lossmult(\bid_{\sf last}) < \bid_{\sf last}$.

\paragraph{\bf  Probabilistic auction duration} The auction terminates at a randomly determined time, namely when the next block is mined. For proof-of-work blockchains such as Ethereum, we model a block interval as a random variable $\blockinterval$, which is exponentially distributed. 

\paragraph{\bf  Rate-limited bidding} A player cannot raise her own bids continuously, but must wait a short interval $\interval$ to do so. This reflects throttling that blockchain peer-to-peer networks perform to prevent flooding attacks. As miners are always free to include the highest paying transaction that they observed, players can only raise bids, not lower them (a natural auction feature that holds in Ethereum PGAs).

\paragraph{\bf Minimum starting bid}
While there is no required minimum gas price in Ethereum, in practice PGA participants will want to start their bid at a price that gives them a good chance of getting included in the next block. Put another way, if one bids too low, then even if nobody else bids against them, a miner may not include the transaction in their block. Thus, although not strictly enforced, we model our auctions as having a minimum starting bid, $s$.

\paragraph{\bf Minimum bid increments}
 Players do not have freedom to raise bids by arbitrary increments. Instead, there is a minimum bid raise, $\iota$, measured as a function of the player's previous bid. This matches Ethereum dynamics in which one can replace their transaction with another one, but it will only be relayed by the peer-to-peer network if the gas price has increased by a minimum threshold.  Parity, the most common Ethereum node software, enforces a default minimum increment of $12.5\%$.

 Importantly, the bid increment restriction for $\player_i$ is entirely a function of the $\player_i$'s previous bid, but there is no enforced relationship between the bids of competing players. That is, players can outbid one another by arbitrarily small quantities so long as they are sufficiently raising their own previous bid.

\vspace{4mm}
\noindent In practice, network latencies and transaction-interval times are stochastic, but we assume constants $\latency_i$ and $\interval$ to simplify our analyses. We believe that this simplification does not significantly affect strategy outcomes. 

Continuous-time, imperfect-information games have seen only limited study. Most similar to our work is {\tt FlipIt}, a continuous-time game with imperfect information that models advanced persistent threats (APTs)~\cite{van2013flipit}. We are unaware of any prior work on such games with random play times or in general with the structure of PGAs.

\subsection{Formal PGA model}

For simplicity, we focus on games where $n=2$, i.e., that involve a pair of players $(\player_0, \player_1)$. Our formalism, though, can be generalized to any $n$. The restriction to two players is supported by our empirical observations, which shows that most PGAs were played out between two players.

We let $\bidtuple = (t, \bid; \pindex)$ denote a bid by player $\player_i$. Here $\timet$ is the time at which the bid is placed, \bid is the bid price, and $\pindex \in \{0,1\}$ is the identity of the bidder.

We denote by $\bidtuples^*$ the sequence of all bids published to the network by all players at the current time $\currenttime{}$, and $\bidtuples$ the full sequence of all bids made by all players up to the point at which the auction ends. We let $\bidtuples[t]$ denote $\bidtuples$ at a particular time $t$. We let $\bidtuples_{i}$ and its variants (e.g., $\bidtuples_{i}[t]$) denote a bid sequence of player $\player_i$. 

A {\em strategy} $\strategy_{\pindex}$ for player $\player_{\pindex}$ is a procedure for participating in a PGA, and may be probabilistic (``mixed,'' in game-theoretic terminology).  $\strategy_{\pindex}$ takes the following form, where input $\currenttime{}$ is the current time and $\playstate_{\pindex}$ is the current local state of $\player_{\pindex}$:

\[(\move,\waketime, \playstate') \sample \strategy_{\pindex}(\bidtuples^*, \currenttime, \playstate).\]

\noindent The output $\move$ is an action by $\player_{\pindex}$. Either $\move = \bidtuple$ for some bid $b$, or else $\move = \bot$, indicating the player is not placing a bid. The output $\playstate_{\pindex}'$ represents an internal state update for the player. 

Finally, $\waketime \geq \currenttime{}$ is a {\em wake time}. This is the time when $\player_{\pindex}$ schedules its next execution {\em assuming no bids appear in the meantime}. The use of wake times means that w.l.o.g., we can assume that a player always emits a bid $\bidtuple = (\timet^*, \bid; \pindex)$, i.e., for immediate publication. We also assume that when executed at an emitted wake time, a player always emits a bid; a player that chooses not to schedule a bid can set $\waketime = \infty$. 

\paragraph{Game execution}
A game between $\player_0$ and $\player_1$ involves an execution $\exec(\strategy_0, \strategy_1, [\blockinterval, \lossmult()])$ of their respective strategies, and accounts for players' imperfect information due to their latencies. We specify the procedure for $\exec$ in Figure~\ref{fig:exec} in Appendix \ref{app:model}.

\paragraph{Payoffs}

Players (bots) compete for financial rewards. We measure these rewards using the standard game-theoretic notion of a {\em payoff}. An execution $\exec(\strategy_0,\latency_0, \strategy_1,\latency_1 [\blockinterval, \lossmult()])$ outputs a pair $(\gain_0, \gain_1)$. These are the respective payoffs (profits or losses) of $\player_0$ and $\player_1$.

We define the payoff of $\strategy_0$ {\em against} $\strategy_1$ as follows:

\begin{align*}
& \payoffexec{\exec}{(\strategy_0,\latency_0) (\strategy_1,\latency_1)}[(\blockinterval, \lossmult())] = \nonumber \\ 
& \hspace{5mm} {\mathbb{E}}[\gain_0 \,\mid\, (\gain_0, \gain_1) \sample \exec(\strategy_0,\latency_0, \strategy_1, \latency_1,[\blockinterval, \lossmult()]).
\end{align*}. 

We refer to a strategy $\strategy_0$ as {\em null-profitable} if $\payoffexec{\exec}{(\strategy_0,\latency_0) (\strategy_{\varnothing},\varnothing}[(\blockinterval, \lossmult())] > 0$ when $\strategy_{\varnothing}$ is the null strategy (i.e., a strategy that never bids and thus its latency is also irrelevant).

\subsection{Why repeated bidding?}\label{sec:whyrepeated}

To understand the significance of our model, it is helpful to see why players' optimal strategies involve placing multiple bids over time. To show why, it is helpful by contrast to observe two settings in which players' optimal strategies involve placing a single bid.

\paragraph{Sealed-bid auction}
Suppose that instead of the complex game just outlined, a PGA were a standard sealed-bid auction (either first or second price, with $\lossmult() = 0$)~\cite{mechanism-design-lecture}. Players bid once and ties are broken randomly.

Players would then be incentivized to cast a single bid as close as possible to 1, throwing away almost the entirety of their potential profit. This behavior is intuitive.

If bids are discrete, with tick sizes (smallest increments) of $\epsilon$, then there is a Nash equilibrium for $S_0 = S_1$, the strategy of bidding $1-\epsilon$. In this case, $\payoffexec{\exec}{\strategy_0, \strategy_1}[(\blockinterval, 0)] = \epsilon /2$. Notice that since only one bid is cast and its sealed, we omit the latency parameter as it is irrelevant.
This is the only Nash equilibrium for deterministic strategies, since a deviation that bids an extra $\epsilon$ would otherwise be profitable. For randomized strategies, non-cooperative players will converge towards this equilibrium after playing multiple games in which they study one another's behavior. 

\paragraph{Fixed block duration}

Returning to our model, when $\latency_i, \latency_j > 0$, placing multiple bids is a useful strategy {\em only when $\blockinterval$ is a random variable}. The reason is that given a block interval of fixed duration $d$, players can place the equivalent of sealed bids in the interval $[d-\interval, d]$ after the mining of the last block. Thus, in this setting the {\em{only}} reason for repeated bidding is on-chain signaling of intent, which is more expensive than out-of-band communication.

\paragraph{Proof-of-stake blockchains} Notably, sealed-bid auctions are likely to be the norm on proof-of-stake blockchains, due to two reasons. First, in many proof-of-stake protocols the identity of near future miners is known in advance (perhaps anonymously, i.e., only the individual miner knows her time slot). This implies that miners can accept bids over secure (encrypted and authenticated) out-of-band channels. Second, the block duration is quite predictable in proof-of-stake chains -- miners forfeit their time slot unless they broadcast their block before a limit (measured according to the local clocks of other miners) is reached. There is some measure of unpredictability: the limit should be generous in order to accommodate propagation delays in the network, and miners may wish to collect many transactions that pay lucrative fees before broadcasting their block. Still, the block duration should behave according to statistical patterns that differ greatly from the stochastic process of PoW blockchains.
Our following analysis may thus apply to proof-of-stake systems with a high enough measure of unpredictability, though sealed-bid auctions are significantly more likely in such systems.

\subsection{What's a good strategy?}

Subject to their individual latencies, all players have identical strategy spaces, and thus we can directly compare strategies while only considering the latencies of the players that executed them. We define

\begin{align*}
    \advantage{\exec}{(\strategy_0, \latency_0),
    (\strategy_1,\latency_1)}[(\blockinterval, \lossmult())] = &\payoffexec{\exec}{(\strategy_0,\latency_0, (\strategy_{1},\latency_1})[(\blockinterval, \lossmult())] - \\ &\payoffexec{\exec}{(\strategy_1,\latency_1), (\strategy_{0},\latency_0)}[(\blockinterval, \lossmult())].
\end{align*}

\section{Priority-Gas-Auction Strategies}

We now report on classes of strategies that we have observed empirically and examine them within the context of our model. Given the large strategy space, we restrict our focus to the most prevalent PGA strategies that we have observed.
 
We thus consider three classes of strategies. The first two strategies are {\em blind raising} and {\em counterbidding}, which are often used in competition with one another. We then consider a {\em cooperative} strategy in which players take turns bidding. While we don't claim that we are seeing perfect cooperation (since if we were there would be no need for on-chain bidding at all), we show the cooperative strategy approximates and predicts a large class of behaviors that we have observed on-chain. Finally, we show that under conditions consistent with the PGAs we are observing on Ethereum, there exist Nash equilibria for cooperative strategies.

\subsection{Latency Wars}
For all strategies that we consider, our model predicts that players with a lower latency will have an advantage. Validating this feature of our model, we have observed empirically that players are consistently and substantially reducing their respective latencies over time (Figure \ref{fig:latency}).

 \begin{figure}
 %\hspace{-8mm}
 \scalebox{0.9}{\includegraphics{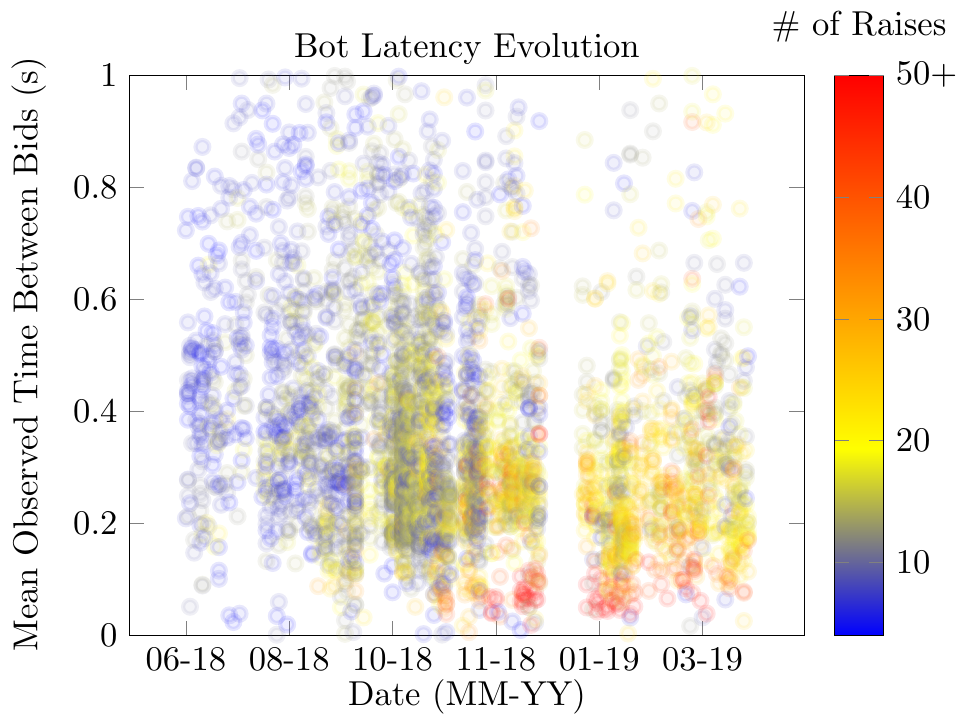}}
\vspace{-6mm}
\caption{Evolution of the PGA bots' latency strategies over time.  Each point is the mean latency of a bot in some observed auction on the plotted date.  The coloring of dots shows the number of total raises each bot placed in that respective auction.  Early in the observed market, the majority of auctions contained 0-10 raises per bot.  Over time, these auctions increasingly contained more raises and a lower mean latency.  We exclude all observed mean latencies that do not fall between 0 and 1 seconds; there are such latencies, which are the result of the limitations of our monitoring architecture mentioned in~\ref{sec:limitations}.  The number of opportunities per bot, however, illustrates a trend that is far less vulnerable to instrumentation limitations, as it does not depend on the order or precise timing of observed transactions.}\label{fig:latency}
\vspace{-4mm}
\end{figure}

\subsection{Blind raising}

{\em Blind raising} is a simple {\em non-adaptive} strategy. Player $\player_i$ raises her own bids under a predetermined schedule, and the strategy is invariant to $\bidtuples^*_{1-i}$, the history of the other player's bids. In PGAs that we have observed, this schedule often involves repeated increases by fixed fractional increments (e.g., 12.5\%, 21\% or 70\%). Basic blind raising is thus also {\em deterministic} (``pure''). A basic version may be modeled as follows:

\vspace{1mm}
\begin{itemize}
    \item []{\em Blind raising ($\strategy_0$):}  $\player_0$ emits a bid with amount $\bid_0$ at time $\tau = 0$ and $\bid_0 \times (1 + f)^k$ at time $k \interval.$
\end{itemize}
\vspace{1mm}

At first glance, this strategy, and 
non-adaptive strategies in general, seem like a bad idea: They fail to exploit information available to the player in $\bidtuples^*$.

Surprisingly, however, the imperfect information created by network latency means that a {\em non-adaptive strategy can achieve an advantage over natural adaptive strategies}. The intuition is this: By playing non-adaptively, a player can publish a bid {\em faster} than by waiting to see the opposing player's bid and reacting to it. We explain below in particular how blind raising can achieve an advantage over a natural, adaptive strategy called {\em reactive} counterbidding.

\subsection{Counterbidding}

{\em Counterbidding} is a strategy in which a player observes an opponent's bidding strategy and reacts by placing higher opposing bids. A simple example of this strategy is what we call {\em reactive counterbidding}, an adaptive strategy that a player $\player_1$ might execute against another player $\player_0$.  The strategy is described formally in Figure \ref{fig:reactive}.

\begin{figure}

\begin{myframe}{{\em Reactive 
counterbidding ($\strategy_1$)  :}}
   \begin{itemize}
    \item At time $\tau = 0$, $\player_1$ bids $s$.
    \item Denote the most recent bid by $\player_1$ as $\bid_1$. $\player_1$ waits until it sees a bid $\bid_0 > \bid_1$ cast by $\player_0$. $\player_1$ immediately counterbids $\min(\max(\bid_1 \times (1+\iota), \bid_1 + \epsilon), 1+\lossmult(\bid_1))$. 
\end{itemize}
\end{myframe}
    \centering
    \caption{Reactive counterbidding PGA strategy. $\epsilon$ is the minimum bid tick, $\iota$ is the minimum bid increment, and $s$ is the minimum starting bid. By bidding $\max(\bid_1 \times (1+\iota)$,  $\player_1$ ensures that its bid meets the minimum bid requirement and is also higher than $\player_0$'s bid. However, it would never make sense for $\player_1$'s bid to exceed $1+\lossmult(\bid_1))$ as at this point it is more profitable to not bid and instead pay $\lossmult(\bid_1)$.}
    \label{fig:reactive}
\end{figure}

Reactive counterbidding can achieve an advantage over a blind raising strategy---{\em provided, as we shall see, that $\latency_1$ is small with respect to $\interval$}. Counterbidding involves outbidding $\player_0$ as quickly as possible by a small amount, provided that the bid is raised by the required minimum.

As $\player_0$ has interval $\geq \interval$ between bids, reactive counterbidding achieves an interval $\geq \interval$ between the bids of $\player_1$.

Many consensus algorithms (e.g., Nakamoto consensus) exhibit $\blockinterval$ exponentially distributed with mean $\lambda$. Under the simplifying assumption that $\lossmult(\bid) = \bidcost$, we can show the following two observations, which we prove in Appendix \ref{app:model}.

\begin{observation}
\label{obs:obs1}
Let $\latency_i = 0$ and $\payoffa = \advantage{\exec}{(\strategy_0,\latency_0), (\strategy_{\varnothing},\varnothing)}[(\blockinterval, \bidcost)]$ be the payoff of null-profitable blind raising strategy $\strategy_0$ against $\strategy_{\varnothing}$. Then for

\begin{equation*}
    \frac{\bidcost}{\payoffa+\bidcost} < e^{-\lambda \interval},
\end{equation*}

\noindent as $\epsilon \rightarrow 0$, reactive counterbidding strategy $\strategy_1$ has positive payoff.
\end{observation} 

We now show, somewhat counterintuitively,  that with the right parameters---specifically, when $\latency_1$ is high w.r.t. $\delta$ ---blind raising has an advantage over reactive counterbidding. 

\begin{observation}[Latency Amplification]
\label{obs:amplification}
For $\latency_1 > \interval$, there exists  a null-profitable, blind raising strategy, $\strategy_0$, such that for any pure reactive counterbidding strategy $\strategy_1$,
%then for any $\blockinterval, \epsilon$ reactive counterbidding strategy 
$\strategy_0$ achieves $\advantage{\exec}{(\strategy_0, \latency_0), (\strategy_1,\latency_1)}[(\blockinterval, \lossmult())] > 0$.
\end{observation} 

\subsection{Cooperation}

We now turn our focus to cooperative strategies, which we believe to be a  close approximation to a common behavioral pattern we observe in which players slowly alternate raising bids. We demonstrate a cooperative Nash Equilibrium for the PGA game that helps shed light on the most common bidding behaviors that we observed. 

Notice that as bids are non-decreasing, with each successive bid that players submit, the 
maximum profitability of the opportunity decreases. Intuitively then, it is beneficial for all players to decrease the total number of bids in the game.

It is natural then to assume that, whenever possible, players will coordinate out-of-band to split the profits rather than decrease their profits by competitively bidding on-chain. However, by the very virtue of us seeing on-chain PGAs, we know that perfect cooperation does not exist. 

For example, if we consider a repeated game in which the same set of players is bidding on multiple (say for simplicity, equal value) opportunities, they can coordinate off-chain so that each opportunity only has a single bid on-chain and the profit for bidders is maximized. Indeed, in our measurement, we found many arbitrage opportunities for which there was only a single player, and it is entirely possible that off-chain coordination is at play here, although by its very nature there is insufficient data on chain to confirm this.

While off-chain cooperation may make sense, particularly in a repeated game, it has some drawbacks including lack of anonymity. We restrict our analysis here to the more interesting case of single game instances where players are bidding on-chain and coordinate their strategies to maximize expected profits Chief among them is that off-chain coordination is cumbersome and requires participants to identify themselves and sacrifice anonymity. This allows us to gain insight into the on-chain competitive auctions that we have observed. 

 As we will see, the cooperative strategies that we study help explain the behaviors in many of the PGAs that we have observed. Moreover, we show that under suitable parameters, there is a Nash Equilibrium for both players to follow a cooperative strategy.

We stress that for our cooperative equilibrium to emerge, it's not necessary that players explicitly coordinate out-of-band, but this behavior can develop organically on-chain. Indeed, we have seen evidence that over time, players have moved closer toward a cooperative equilibrium, which is consistent with well known results in experimental game theory that participants in the wild will converge to an equilibrium over time \cite{nagel1995unraveling,cabrales2007equilibrium}.

\paragraph{Grim trigger equilibrium}

We now describe a particular cooperative strategy in which any defection is responded to by immediately bidding the maximum amount, thereby eliminating all profitability from the auction. 

This strategy class is similar to the {\em grim trigger} strategy that appears in the game theory literature in the context of repeated games \cite{prisoners-dilemma-strategies}. Notice that since PGAs involve successive opportunities to bid potentially in response to the bids of other players, even a single iteration of our auction has elements of a repeated game.

Consider two players who coordinate with the following strategies: at set intervals, players alternate their bids, each time bidding the minimum increase. If any player deviates either by bidding too soon or by raising the bid higher than they are supposed to, the other player will respond by raising their bid such that the game is no longer profitable. This strategy is characterized by following parameters: $V$ is the set of agreed upon times at which bids will be scheduled, and $W$ is the set of agreed upon bid values, such that that $W[i]$ will be bid by the appropriate player at time $V[i]$.

The grim-trigger cooperative strategy is formally described in Figure \ref{fig:cooperate}.

\begin{figure}

\begin{myframe}{{\em Cooperative ($\strategy_i$) with parameters $D,W$:}}
    \begin{itemize}
    \item 
    At time $\tau = V[2k + i]$, $\player_{i}$ emits a bid with amount $W[2k + i]$.
    \item If $\player_i$ observes a bid of value \bid at time $t$ such that $t < V[2k+i]$ but $\bid > W[2k+i -1]$, then $\player_i$ immediately bids $\$1+\lossmult(\bid)$.
\end{itemize}
\end{myframe}
    \centering
    \caption{Grim-trigger cooperative PGA strategy}
    \label{fig:cooperate}
\end{figure}

Notice that in the second condition, we assume that $\player_i$ can determine the time $t$ that the bid was emitted by $\player_{1-i}$. This is consistent with our model in which the latency of each player is known. Thus, even though $\player_i$ will not observe $\player_{1-i}$'s bid until time $t + \latency_i$, it can compute $t$ by subtracting $\latency_i$ from the time that it first observed the bid. 

As we will see, the knowledge gap due to latency is an important feature when determining if our model has an equilibrium. All defections will eventually be detected, but depending on the latencies, there may still be sufficient time to profit from deviating.

\paragraph{Minimum bid raises}

In a cooperative PGA, the main function of bidding on-chain is to alternate bids between players, but all players are incentivized to keep the actual bid price as low as possible to maximize profit. The following observation follows:

\begin{observation}
\label{obs:optimalbids}
In a grim-trigger cooperative PGA, for all $i \leq i_{max}$, the optimal choice for bids is $W[0] = s$ and $W[i] =W[i-1]\times (1+\iota)$.
\end{observation}

Remarkably as shown in Figure \ref{fig:raisestrategyscatter}, in our observed PGAs, players converge over time to a minimum bid raise of $12.5\%$, which is the minimum allowed raise in the Parity client. It is known that players in the wild will converge to Nash Equilibria \cite{cabrales2007equilibrium,nagel1995unraveling} and this on-chain behavior is consistent to a convergence toward elements of cooperative equilibrium.

We stress that we don't claim that perfect cooperation is occurring or even that our exact equilibrium is on play. We do claim, however, that players are converging to a more cooperative state, that is approximated by our model, in which they allow opportunities for other players in an effort to maximize their own expected profit over time.

 \begin{figure}
 \scalebox{0.9}{ \includegraphics{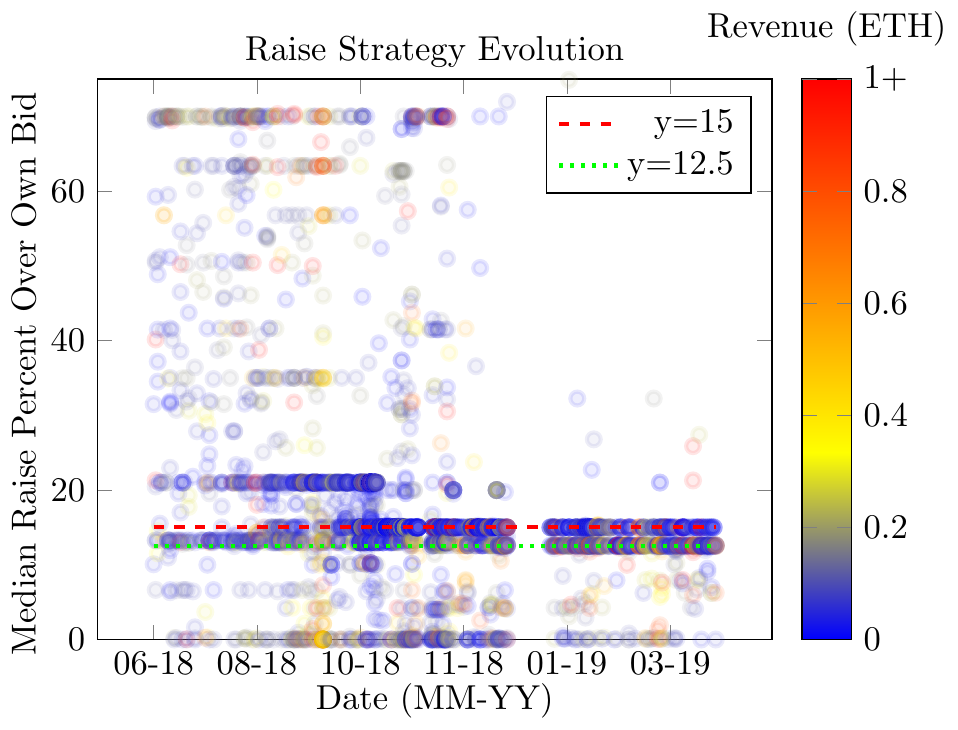}}
\vspace{-6mm}
\caption{Evolution of the PGA bots' median raise strategies over time for bots with more than 4 bids in pure revenue PGAs.  Note the strategy convergence in the iterated game to almost exclusively the 12.5 (minimum raise) and $15\%$ (deviation from minimum raise) levels.  Opportunities below the $y=12.5$ line are likely artifacts of bots emitting multiple transactions from globally distributed nodes; these opportunities correspondingly show a bias towards high revenue compared to those using the predicted strategies.  We trim outliers and omit points with median raises over 75\%; we observe 5 such high-raise outliers of 4,007 total observations.}\label{fig:raisestrategyscatter}
\vspace{-4mm}
\end{figure}

\paragraph{Nash equilibrium} We now present our main result showing that there exist Nash equlibria for a wide range of parameters using the grim-trigger cooperative strategy.

\begin{theorem}
%\begin{observation}
\label{obs:main}
For parameters consistent with Ethereum PGAs, there exist a grim-trigger Nash equilibria for a 2 player PGA where both players follow a $D,W$ cooperative strategy.
\label{obs:ethparams}
%\end{observation}
\end{theorem}

%Towards showing this, we begin with several observations.

\begin{figure}
 \scalebox{0.5}{\includegraphics{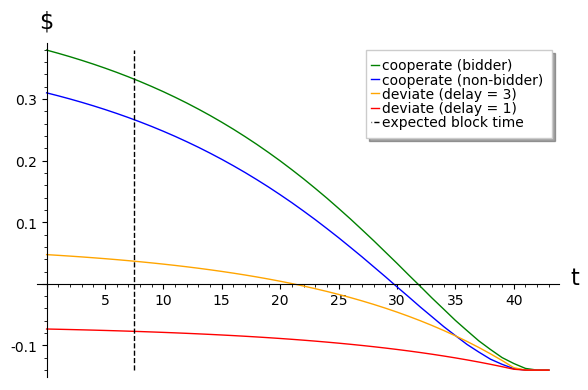}}
 \caption{Cooperative PGA with interval $t=2$}
 \label{fig:model1}
\end{figure}

\begin{figure}
 \scalebox{0.5}{\includegraphics{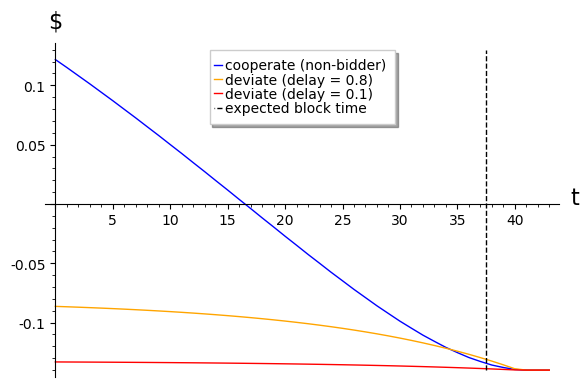}}
 \caption{Cooperative PGA with interval $t=0.4$}
 \label{fig:model4}
\end{figure}

Figure~\ref{fig:model1} and Figure~\ref{fig:model4}  show the expected payoff of the bidder and non-bidder in a cooperative PGA. The x-axis indicates the starting time of the interval, where the expected duration of the game is $1/\lambda = 15$ seconds. The y-axis indicates the expected payoff.
Note that the players are alternating: if Player~1 is the bidder and Player~2 is the non-bidder when the first interval starts, then Player~1 is the non-bidder and Player~2 is the bidder when the second interval starts. 
These figures also display the profitability of deviation, i.e., the expected payoff of the non-bidder (according to latency parameters) in case she deviates when the interval begins.

If the deviation payoff surpasses the non-bidder payoff at any point (as with the
$\texttt{delay=3}$ plot in Figure~\ref{fig:model1}), then a Nash equalibrium cannot hold for those parameters -- a backward induction argument implies that if it is profitable for the non-bidder (say, Player~1) of interval $i$ to deviate, then it must also be profitable for the non-bidder (Player~2) of interval $i-1$ to deviate, and so on.

Figure~\ref{fig:model4} corresponds to the Ethereum parameters per our empirical measurements: $t=0.4$ seconds is the mean interval length and $\delta=0.1$ seconds is the estimated latency in our measurements.

\smallskip
To prove our main claim, let us recall that each player's bid increments her previous bid by $12.5\%$. Given our measurements, this implies $17$ rounds before the players give the entire auction's reward to the miners. We can thus verify that in each of the few dozens of intervals prior to the end of the game, the expected payoff of adhering to the cooperative strategy is larger than the expected payoff upon deviation.
Indeed, the Ethereum parameters that are plotted in Figure~\ref{fig:model4} show that the expected non-bidder payoff is always greater than the expected payoff of deviation. Theorem~\ref{obs:ethparams} thus follows. 
\section{Miner-Extractable Value and Blockchain Security}
\label{sec:securityissues}

PGAs and DEX arbitrage may not seem immediately harmful or relevant to the security of an underlying blockchain. They might simply seem an efficient mechanism for conveying market information among network participants.  Unfortunately, we now argue, DEXes in fact present a serious security risk to the blockchain systems on which they operate, i.e., {\em at the consensus layer}. In other words, a key result of our work is that {\em application-layer security poses a current and direct threat to consensus-layer security}.

In a stable blockchain, block rewards incentivize honest miner behavior. As we show empirically in this section, though,
\emph{order optimization} (OO) fees, or implicit fees a miner is able to reap by leveraging their control of a consensus epoch, can exceed the block reward and instead incentivize forking attacks.  To capture OO fees, a miner can reorder users' transactions and potentially insert their own, reaping profit in Ether directly to their account.  For example, a miner could execute the pure revenue transactions described in this work themselves, while still claiming the PGA fees for all the ``losing" arbitrage bots attempting to do the same. 

OO fees represent one case of a more general quantifiable value we call {\em miner-extractable value} (MEV).  MEV refers to the total amount of Ether miners can extract from manipulation of transactions within a given timeframe, which may include multiple blocks' worth of transactions.  In systems with high MEV, the profit available from optimizing for MEV extraction can subsidize forking attacks of two different forms. The first is a previously shown {\em undercutting attack}~\cite{carlsten2016instability} that forks a block with significant MEV. The second is a novel attack, called a {\em time-bandit attack}, that forks the blockchain retroactively based on past MEV.

Undercutting attacks were previously considered a risk primarily in the distant future, when block rewards in Bitcoin are expected diminish below transaction fees. By measuring the significance and magnitude of OO fees, our work shows that undercutting attacks are a {\em present threat}.

Time-bandit attacks are also a present and even larger threat. They can leverage not just OO fees, but {\em any} forms of miner-extractable value obtained by {\em rewinding} a blockchain.
Time-bandit attacks' existence implies that DEXes and many other contracts are inherent threats to the stability of PoW blockchains, and the larger they grow, the bigger the threat. 

\subsection{OO fees: Measurement study}

We now lower bound the severity of OO fees in DEXes.

Figure~\ref{fig:oodistribution} shows the distribution of pure revenue OO fees as a fraction of total miner-extractable value in Ethereum.  We conservatively estimate MEV in this context as the sum of explicit transaction fees and pure revenue OO fees, making our distribution an underestimate of the share of MEV that does not come from explicit fee payments in Ethereum. Note that PGAs themselves are an indirect mechanism for miners to claim OO fees: if miners instead claimed OO fees directly, there would be no incentive to bid in or create a PGA.

Of all blocks since block 3,875,490 (the first block we observed with a pure revenue OO fee from decentralized exchange arbitrage), approximately 3.6\% contain at least one pure revenue arbitrage transaction.  Of these blocks, 17,897 blocks contain pure revenue OO profits that are $<1\%$ of the total stealable fees, confirming our earlier conclusions that the most frequent pure revenue arbitrage opportunities are for small price differences in the course of normal trading across exchanges.  Nonetheless, many blocks do contain substantial pure revenue fees.  36,860 observed blocks contain pure revenue fees that exceed 20\% of total stealable fees, 12,002 blocks derive the majority (over 50\%) of stealable fees from pure revenue, and 2,517 blocks derive more than 80\% of stealable fees from pure revenue.

\begin{figure}
\scalebox{0.9}{\includegraphics{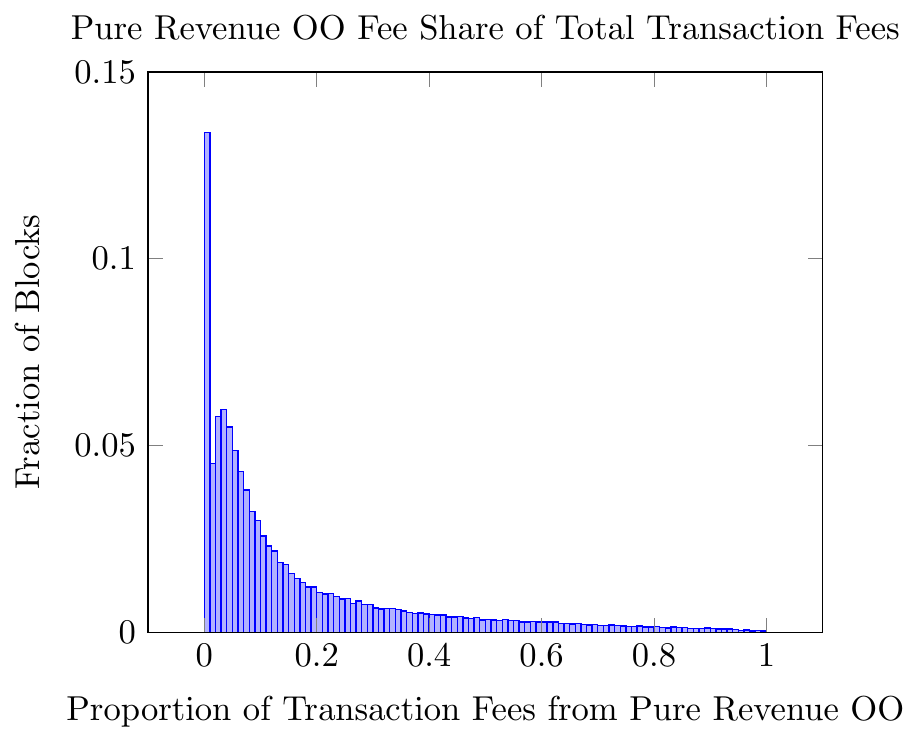}}
\vspace{-2mm}
\caption{Distribution of pure revenue OO as percentage of block transaction fees in blocks containing such transactions on the Ethereum network, blocks 3875490 to 7408826; 133,815 of 3,533,336 contain observed pure revenue (3.6\%).  Since block 7,000,000 (Jan 2 2019), this has increased to 6.3\%.}\label{fig:oodistribution}
%133815 7408826 3875490
\end{figure}

\begin{figure}
\includegraphics{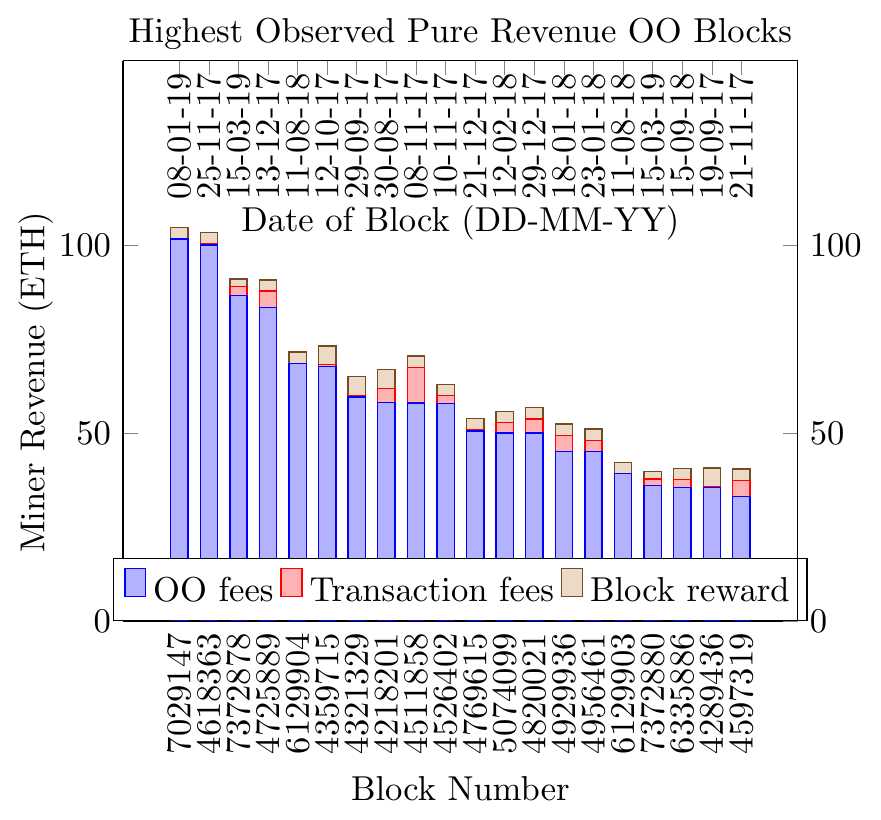}
\vspace{-5mm}
\caption{Blocks with the highest pure revenue OO fees observed on Ethereum.  As is shown here, OO fees in these blocks dominate both block rewards and transaction fees, often by more than an order of magnitude.}\label{fig:hotoo}
%133815 7408826 3875490
\vspace{-4mm}
\end{figure}

Occasionally, these fees can be substantial and provide substantial miner incentives to orphan blocks or otherwise deviate from the mining protocol. Figure~\ref{fig:hotoo} shows the 20 blocks observed on Ethereum with the highest absolute pure revenue OO fees.  

As an example, the highest block, block 7029147 (\url{https://etherscan.io/block/7029147}), contained an arbitrage trade generating a revenue of 101.6 ETH, dwarfing the block reward of 3 ETH and the insignificant explicit transaction fees of 0.022 ETH. In this particular block, one transaction generated all the pure revenue OO fees.  This transaction, whose profit graph is available at \small \url{https://bit.ly/2D4U57c}\normalsize, saw an exchange of Bigbom token (BBO) for ETH on Bancor and Kyber; an error in Bancor's pricing formula allowed a user to buy BBO at rates far under market. The bot's order size was 101.6 ETH, of which almost all was taken as pure profit. 

\emph{As with all our analyses, we stress that our measurements are conservative: They represent lower bounds on arbitrage behavior.}  The limitations in our instrumentation harness discussed in Section~\ref{sec:limitations}, lead to potential underestimation.

\subsection{Undercutting attacks}

Undercutting attacks~\cite{carlsten2016instability} represent one vector of attack that can leverage OO fees.

It is well known that fixed per-block miner rewards are a key feature of secure and stable cryptocurrency protocols, as originally described in~\cite{carlsten2016instability}.  In that work, the authors analyze the incentives of blockchain miners in a regime where transaction fees exceed the inflationary subsidy paid to miners by a blockchain protocol. They observe that when the block reward is dominated by fees, rewards have high variance. As a result, a miner can fork a high-fee block, holding back some fees to attract other miners to build on the fork.  In extreme cases, incentives to deviate from the protocol may lead to disruption in miner strategies for economically rational miners, reducing the security provided by block confirmations.

Prior to our work, transaction fees have been viewed as the only source of value for undercutting attacks, with~\cite{carlsten2016instability} even calling a world in which fees dominate block rewards a ``transaction-fee regime.'' That regime seemed a distant prospect. The title of~\cite{carlsten2016instability}, ``On the Instability of Bitcoin Without the Block Reward,'' refers to Bitcoin block rewards going to zero, an event anticipated around the year 2140.

Our study shows that OO fees are a form of value that sometimes dominates explicit transaction fees {\em today}. The tail of the distribution in Figure~\ref{fig:oodistribution} (transaction fee proportions $> 0.5$) and all of the example blocks in Figure~\ref{fig:hotoo} represent such opportunities. 
In other words, undercutting attacks represent a {\em present threat} in Ethereum, and one that will grow with the success of smart contracts that attract OO fees.  Other potential sources of OO fees are mentioned in~\cite{clark-prediction}.

Most importantly, our view of arbitrage remains conservative, and we emphasize that pure revenue opportunities represent a small slice of total ordering fees. There are \emph{many} possible sources of ordering fees payable to miners, including more sophisticated arbitrage. There are also many other forms of miner-extractable value. Miners can ``steal'' arbitrage opportunities from arbitragers by taking them themselves. Additionally, the survey in~\cite{eskandari2019sok} describes several other sources of MEV, including the ability to buy into profitable ICOs early and the manipulation of games of chance.

Undercutting attacks can leverage OO fees or any other MEV from newly generated blocks. The second attack we describe can use MEV from {\em past blocks} to subsidize an attack.

\subsection{Time-bandit attacks}

We now describe {\em time-bandit attacks}, a new, second attack vector that can exploit miner-extractable value. They can exploit MEV from new blocks, but more powerfully, can also use MEV from {\em past blocks} via rewinding.

Time-bandit attacks are conceptually simple. Suppose a blockchain has a suffix (subchain) $[{\sf height}_0,{\sf height}_1]$, with current block height ${\sf height}_1$, in which stealable value exceeds block rewards. An adversary can rewind to ${\sf height_0}$ and use the resulting MEV to subsidize a profitable 51\% attack that mines a fork up to or past ${\sf height}_1$.

Of course, a time-bandit attack relies on real-time access to massive mining resources. As noted in~\cite{bonneau2018hostile}, however, ``rental attacks'' are feasible using cloud resources, particularly for systems such as Ethereum that rely heavily on GPUs, which are standard cloud commodities. Sites such as http://crypto51.app/ estimate the costs.  We illustrate with an example that relies on MEV from DEX rewinding.

\begin{example}
Consider a price spike from 1 USD to 3 USD in a token that trades on an on-chain automated market maker (e.g., Bancor~\cite{Bancor:2019}). A miner performing a time-bandit attack can now rewrite history such that it is on the buy side of every trade, accruing a substantial balance in such a token at below market rate---{\em before the price spike}.\footnote{The attacker can additionally use new user trades at the market 3 USD price in the automated market maker system to offload their tokens into ETH.}

For example, if the attacker wishes to rewrite 24 hours of history, and 1M USD of volume occurred in exchanges with rewritable history in that token, then the attacker can obtain a MEV gross profit of 1 M $\times$ (3 USD - 1 USD) = 2M USD.\footnote{More sophisticated reordering strategies stand to profit the miner even more. The miner can selectively include orders that manipulate the prices received by other traders \emph{in the past}.}

At the time of writing (March 2019), http://crypto51.app/ estimates a 24-hour 51\% rental-attack cost on Ethereum of about 1.78M USD, implying a net profit of around 220K USD.
\end{example}

We stress that time bandit attacks are \emph{not limited to MEV from DEXes}.  A variety of smart contract systems allow anyone to participate and earn profit for doing so, often a desirable design goal for the style of permissionless and open interaction that lends itself naturally to blockchains.  Any on-chain action in the past that could have potentially profited a miner today, including actions that unconditionally earn them ETH in the past, are thus potential sources for time-bandit attacks.  Because smart contracts are Turing complete scripts and carry complex interactions, estimating the size of these opportunities is a challenging problem we defer to future work.

\vspace{2mm}
\noindent \emph{Time-bandit attacks in Ethereum:} Recent transaction statistics suggest that Ethereum is vulnerable to time-bandit attacks. For example, decentralized exchange volumes show a peak of 1.5 billion USD of traded assets on Ethereum's decentralized exchanges in July 2018~\cite{dexvolume}.  While it is hard to gauge the total stealable value in this volume, the current estimated one-month cost of a 51\% attack on Ethereum according to \url{http://crypto51.app/} is approximately 56 million USD, more than {\em 25 times lower than this DEX volume}.  

We posit that the OO fees alone that we have described threaten the security of today's Ethereum network. As Figures~\ref{fig:oodistribution} and~\ref{fig:hotoo} show, blocks with high OO fees and/or arbitrage opportunities can already enable such attacks.

More generally and alarmingly, time-bandit attacks can be subsidized by a malicious miner's ability to rewrite profitable trades retroactively, stealing profits from arbitrageurs and users while still claiming gas fees on failed transactions that attempt execution.  The resulting MEV is potentially massive, suggesting a possibly serious threat \emph{in Ethereum today}.

Of course, a full analysis of the threat would require an understanding of how time-bandit attackers might compete against one another to harvest MEV---by analogy with PGAs. This is a topic for future research. 

\section{Open Questions and Future Work}
\label{sec:futurework}

Our results raise many important questions, some about the arbitrage community itself. For example, it would benefit arbitrageurs to collude with miners, but we observe no such collusion: Preliminary experiments show that bot transactions are equally distributed across mining pools. Are there incentives to avoid collusion, such as concern about the exogenous impact of miner malfeasance coming to light? 

Other questions arise from PGA modeling. Are PGAs positive- or negative-sum games? In what ways could our proposed model be helpfully enriched? 

More broadly, it is important to observe that DEXes are just the tip of the iceberg. At the time of writing, IDEX, the largest, is ranked \#119 by volume by \url{coinmarketcap.com}. It has about 1M USD in 24h volume, compared with 970M USD for Binance, the leading centralized exchange. DEX volume is roughly 0.01\% that on centralized exchanges.

Malfeasance in centralized exchanges might well be rampant---possibly even more egregiously so than in DEXes~\cite{Bitfinex:2018}. As activity in such exchanges takes place off-chain, it is private, and cannot easily be ascertained without privileged access. Additionally, while centralized exchange malfeasance might seem not to impact on-chain security, the two are inextricably linked. By altering token trading dynamics on-chain, an adversary could manipulate and profit from centralized exchange dynamics. For example, a time-bandit attack could accumulate cheap tokens for offloading in a centralized exchange. 

These observations raise several key follow-up questions:

\begin{itemize}
    \item What forms do arbitrage take in centralized exchanges and at what volumes? And frontrunning?
    \item Barring direct data access, what techniques can be used to accurately measure various forms of trading activity in centralized exchanges?
    \item What financial incentives do centralized exchanges create for malfeasances in DEXes? How might they impact blockchain stability?
    \item What other insights might our data yield on DEX arbitrage, especially in quantifying opportunities that are not pure revenue?  How much larger is the full arbitrage economy than executed? Can tight bounds be developed for the amount of MEV on Ethereum today?
\end{itemize}

\section{Conclusion}

We have reported on a sizable economy of bots profiting from opportunities provided by transaction ordering in DEXes.  We quantify the breadth of a specific subset of arbitrage, \emph{pure revenue opportunities}, providing lower bounds on the profitability of ordering manipulation.

We have also formally modeled the behavior of bots competing against each other for miner-supplied transaction priority in \emph{priority gas auctions}. Our empirical study validates several key predictions of our model, including the convergence of bots on a form of profitable cooperation involving minimal gas-price increases. We also show that in many concrete cases, bots' revenue from pure revenue arbitrage alone far exceeds the Ethereum block reward and transaction fees.

Finally, we argue that \emph{miner extractable value}, particularly in the form of \emph{order optimization fees}, implicit fees from modifying transaction order, threatens blockchain consensus stability. Such fees are large enough to subsidize serious attacks on the network. They constitute an economic vulnerability that should be a current cause for concern in Ethereum.

\section*{Acknowledgments}

We would like to thank Matt Weinberg for productive discussions surrounding our model of priority gas auctions.

\vspace{2mm}

This material is based upon work supported by the National Science Foundation Graduate Research Fellowship under Grant No. DGE-1650441, as well as by NSF grants CNS-1514163, CNS-1564102, and CNS-1704615, as well as by ARO grant W911NF16-1-0145.

We also thank IC3 industry partners for their support in funding this work. 
\printbibliography

\appendices
\section{Smart-Contract-Enabled Complex Nondeterminism}
\label{sec:complexnondeterminism}

Although this work often focuses on pure revenue opportunities, not all opportunities are pure revenue.

While arbitrage bots' intent is often clear by inspection of their transactions, other times, intent may be opaque, nondeterministic, or may depend on unpredictable aspects of the Ethereum network state.  The smart contract wrappers through which bots can execute transactions enable more than simple batching of transactions.  Such contracts are actually atomically-executing programs written in Turing-complete scripting languages, and can therefore implement complex strategies.

For example, during, before, or after executing a set of trades, a wrapper contract can also query prices from and conditionally trade with automated market makers.  Because the prices offered by these automated market makers cannot be known in advance, and because batches can revert, even based on changes in the Ethereum global state and the success of other batches, orders placed through smart contracts can express complex  conditional preferences on trade execution. Complex trades can also perform arbitrary computation on the Ethereum network.  Such complex order types do not have analogs in traditional HFT, but are similar to (and more general than) proposals such as~\cite{sandholm2002winner}.

\section{Additional Data Analyses}

\subsection{Pure Revenue USD Market}
\label{sec:purerevenueusd}

\begin{figure}
\hspace{-4mm}
\scalebox{1.0}{\includegraphics{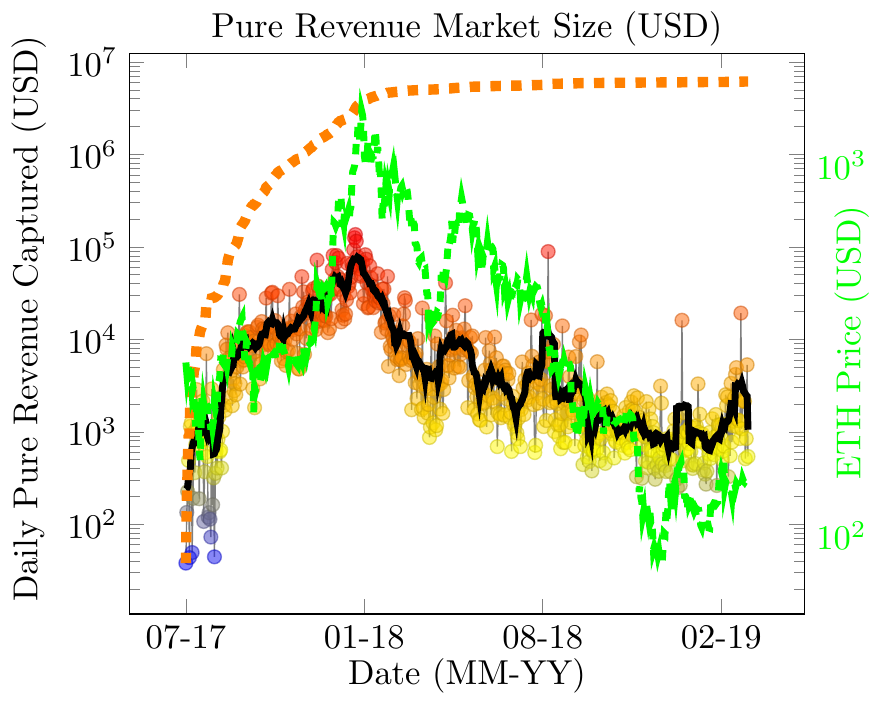}}
\vspace{-2mm}
\caption{Size of the pure revenue market in USD as shown in Figure~\ref{fig:purerevenueethscatter} for ETH, displayed with USD/ETH price overlay.}\label{fig:purerevenueusdscatter}
%133815 7408826 3875490
\end{figure}

Figure~\ref{fig:purerevenueusdscatter} shows the ETH price and size of the pure revenue market in USD.  When compared to the equivalent graph denominated in ETH (Figure~\ref{fig:purerevenueethscatter}), some effect of price reductions on the market size are evident.  Opportunities in ETH are relatively consistent and well distributed, as we expect from mechanical rents extracted due to exchange design fundamentals.  Such USD graphs appear more correlated with the price, showing decreased revenue during price slumps and contradicting our other data about the sophistication of the PGA bot market increasing over time.

Because of this and the relevance of ETH to direct protocol security analyses, we denominate the majority of our synthesized results in ETH.

\subsection{Pure Revenue, Broken Down by Bots}

\begin{figure*}
\scalebox{1.0}{\includegraphics{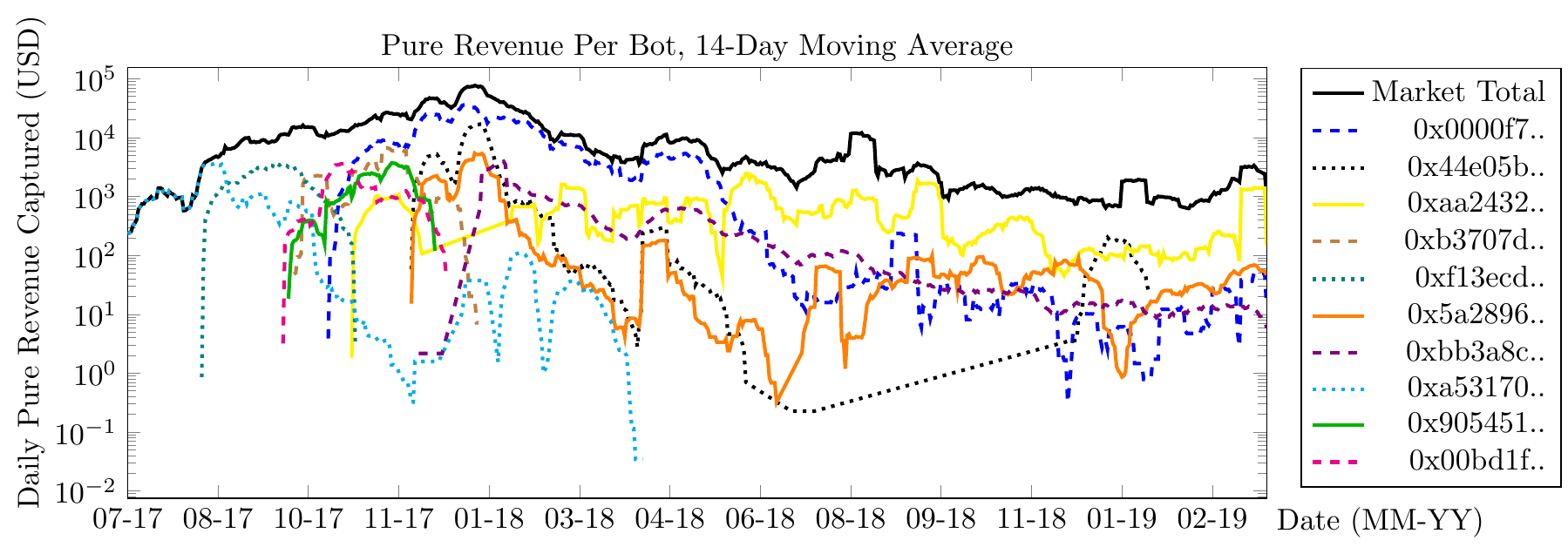}}
\vspace{-2mm}
\caption{Pure revenue bot breakdown, as described in Section~\ref{sec:prevalence}, showing revenue without subtracting transaction costs.}\label{fig:purerevenuebotrevenue}
\end{figure*}

Figure~\ref{fig:purerevenuebotrevenue} shows the size of the full pure revenue opportunity market captured by the top 10 transaction senders; note that the graph is very similar to the graph of the estimated profits we calculate, having a relatively constant offset and almost identical shape for all market players.

This may reflect imperfections in our heuristics for calculating profit; we cannot include, for example, server costs, and it is difficult to tell which on-chain transactions other than a pure revenue transaction may have contributed to a bot's specific costs.  We also do not include costs of failing to capture competitive pure revenue opportunities, which may be significant; it is difficult to disambiguate these failed transactions from unrelated transactions, as they may never execute or attempt to execute the intended trades.  These matters are left to future work and analyses of the data we publish.

\section{PGAs and their participants' games}
\label{sec:intuition}

To provide some informal intuition for the modeling we explore in~\ref{sec:model} we now describe the perspective of both bots and miners participating in PGAs.  This also provides the basis for our later claims that PGAs have the ability to lead to systemic instability in blockchains, degrading their security, a notion we explore fully in Section~\ref{sec:securityissues}.

\subsection{The miner perspective}

From the point of view of the miner who is eventually going to mine a block, the economic game proceeds as follows.  At some point in time (along the approximate objective time axis displayed in Figure~\ref{fig:exampleauction}), the miner takes a snapshot of all transactions it has seen on the peer to peer network, and includes the highest fee transactions in a block template that it attempts to solve the proof of work on, converting the block template into a block able to be propagated to the network.

When this snapshot is taken, the highest $n$ transactions generated by PGA bots  are included in the block.  $n$ is the number of unique (account, nonce) pairs observed in the epoch and therefore the number of unique eligible bidders: recall each (account, nonce) can only be mined once.  The highest paying transaction is included first, and captures the associated revenue of the opportunity.  The remaining transactions are included in descending order of gas price, or marginal cost per instruction.

If all ``bids" are competing for the same pure revenue opportunity, transactions after the first receive no revenue.  If such bids are complex probabilistic orders which are not inherently mutually exclusive, we may expect to see decreasing power-law profits, the lions share of which go to the first arbitrageur.

Note that in Figure~\ref{fig:exampleauction}, the mined block template was likely formed immediately after the winning transaction and before the next transaction with the same (account, nonce) pair; otherwise, the higher fee transaction would have been included by the miner to maximize revenue.  This means that the block template that would finally win this game was created by the miner around second 5. Note also that arbitrageurs continue to bid after this result is locked in, unaware that this pool has formed the template that will eventually constitute a valid block, showing their imperfect information with regards to this game.  After some delay, the bots learn of the outcome, and stop bidding in the auction.

This picture overview clearly that it is in the miner's best interest to refresh the block templates it is mining on as quickly as possible.  Let us assume the real block was mined after the end of the auction, around second 15.  This assumption is reasonable, as there is no incentive for the bots we observe to continue bidding after they are made aware a block is discovered, and there is also substantial incentive for miners to optimize propagation latency to the network (see e.g.~\cite{eyal2014majority}.  Had the miners instead included the transactions that bid highest in this action, their total capture profit would be $231520 \cdot 8856.24 \text{ Gwei}=x \text{ ETH}$ for bot 0xb8D76f4BC2518F8eb508bf0Ccca76f8F9DD57a3f, and $227534 \cdot 7716.48 \text{ Gwei}=y \text{ ETH}$ for bot 0x6BEcAb24Ed88Ec13D0A18f20e7dC5E4d5b146542, for a 
total of $z$ ETH gained by an optimal miner. Assuming the third bot would have bid similarly to these two bots and would have contributed equally to the miner's reward, MiningPoolHub's suboptimality in this instance cost $q$ ETH.

It is obvious to see why substantial arbitrage opportunities then directly alter miner strategies, by incentivizing them to optimize on \emph{template refresh latency} as described above.  If PGA opportunities increase, we expect to see such optimizations in practice.

\subsection{The PGA bot perspective}

From the perspective of the bots competing for block space in PGAs, the game looks very different.  Bots, like miners, see transactions on the public peer to peer network performed by their peers.  They are able to simulate the execution of these transactions, read the orderbooks of decentralized exchanges, and issue orders.  Bots should optimize on latency to other bots to gain information more quickly about their actions.  Bots should also optimize on latency to miners, to make it more likely that at the time when the miners form the template, their bid was the last bid received by the miner and therefore reflects the most market information.  We explore the full space of bot strategies in Section~\ref{sec:model}.  Our model shows that optimizing for latency is profitable for a bot participating in this market, and draws the connection to high frequency (latency-sensitive) trading in traditional markets.

\section{Additional PGA modeling details}
\label{app:model}

\begin{figure*}[ht]
\begin{center}
\fbox{
\procedure[linenumbering]{$\exec(\strategy_0, \latency_0, \strategy_1, \latency_1, [\blockinterval, \lossmult()])$}{
\currentbidtuples, \pend, \delayed{0},\delayed{1} \leftarrow \emptyset; \currenttime{}, \currenttime{0}, \currenttime{1} \leftarrow 0 \\
\timet_{{\sf end}} \sample \blockinterval\\
\hspace{-1mm}\pcdo\\
    \pcind \pcif \firstbidtime(\pend) \leq \min (\currenttime{0}, \currenttime{1},\firsttime(\delayed{0}),\firsttime(\delayed{1})) \pcthen\\
    \pcind \pcind \currenttime{} \leftarrow \firstbidtime(\pend)\\
        \pcind \pcind \bidtuple =  (\hat{t}, \bid; \pindex) \leftarrow \popfirst(\pend) \\
        \pcind \pcind \currentbidtuples \leftarrow \currentbidtuples \,\cup\, \bidtuple \\
        \pcind \pcind \delayed{1-\pindex} \leftarrow \delayed{1-\pindex} \cup (\hat{t} + \latency_{1-\pindex}) \\
        % \pcind \pcind \pcfor \pindex \in \{0,1\} \pcdo\\
        %     \pcind \pcind \pcind (\move, \playstate_i', \waketime{,i}) \sample \strategy_{\pindex}(\currentbidtuples, \playstate_{\pindex}, \currenttime{})\\
        %     \pcind \pcind \pcind \playstate_{\pindex} \leftarrow \playstate'; \currenttime{\pindex} \leftarrow \waketime{,i}\\
        %     \pcind \pcind \pcind \pcif \move \neq \bot \pcthen\\
        %     \pcind \pcind \pcind \pcind \pend \leftarrow \pend \,\cup\, \move\\
\pcind \pcif \text{for } {\pindex} \in \{0,1\}, \firsttime(\delayed{i}) < \firstbidtime(\pend) \text{ and } \firsttime(\delayed{\pindex}) \leq \min (\currenttime{0}, \currenttime{1},\firsttime(\delayed{1-\pindex}))  \pcthen\\
\pcind \pcind \currenttime{} \leftarrow \poptime(\delayed{i})\\
        \pcind \pcind (\move, \playstate_i', \waketime{,i}) \sample \strategy_{\pindex}(\currentbidtuples[\currenttime{}-\latency_{i}], \playstate_{\pindex}, \currenttime{})\\
        \pcind \pcind \playstate_{\pindex} \leftarrow \playstate_{\pindex}'; \currenttime{\pindex} \leftarrow \waketime{,i}\\
        \pcind \pcind \pcif \move \neq \bot \pcthen\\
        \pcind \pcind \pcind \pend \leftarrow \pend \,\cup\, \move\\
\pcind \pcif \text{for } {\pindex} \in \{0,1\}, \currenttime{\pindex} < \min (\firstbidtime(\pend),\firsttime(\delayed{0}),\firsttime(\delayed{1}))  \text{ and } \currenttime{\pindex} \leq \currenttime{1 - \pindex}  \pcthen\\
        \pcind \pcind \currenttime{} \leftarrow \currenttime{\pindex} \\
        \pcind \pcind (\move, \playstate_i', \waketime{,i}) \sample \strategy_{\pindex}(\currentbidtuples[\currenttime{}-\latency_{i}], \playstate_{\pindex}, \currenttime{})\\
        \pcind \pcind \playstate_{\pindex} \leftarrow \playstate_{\pindex}'; \currenttime{\pindex} \leftarrow \waketime{,i}\\
        \pcind \pcind \pend \leftarrow \pend \,\cup\, \move\\
\pcuntil \firstbidtime(\pend) > \timet_{{\sf end}}\\
\bidtuples \leftarrow \currentbidtuples\\
\move_{{\sf win}} =  (\timet_{{\sf win}}, \bid_{{\sf win}}; \pindex) \leftarrow \maxbid(\bidtuples)\\
\move_{{\sf lose}} =  (\timet_{{\sf lose}}, \bid_{{\sf lose}}; 1-\pindex) \leftarrow \maxbid(\bidtuples_{1-\pindex})\\
\pcif \pindex = 0 \\
\pcind {\sf output}\, (\gain_0, \gain_1) = (\payoff - \bid_{{\sf win}}, \lossmult(\bid_{{\sf lose}}))\\
\pcelse\\
\pcind {\sf output}\, (\gain_0, \gain_1) = (\lossmult(\bid_{{\sf lose}}),\payoff - \bid_{{\sf win}})
}
   } 
\end{center}
\caption{{\bf Pseudocode for $\exec$, the execution of a PGA.} This procedure maintains a list of pending bids $\pend$ awaiting transmission. Bids in the pending list, $\pend$, are added to the current bid list, $\currentbidtuples$, at their scheduled time, but each player will not see them until after their individual latency, $\latency_i$. For each bid that is cast by player $\player_{\pindex}$, a time is added to $\delayed{1-\pindex}$ that indicates when $\player_{1-\pindex}$ will first see the bid. Players are awoken to schedule bids at their specified wake time, or earlier if they see a bid cast by the other player.
For clarity, we define several functions. $\firstbidtime(\pend)$ is the time $\timet$ associated with the earliest pending transaction.  $\popfirst(\pend)$ pops and outputs the earliest transaction; given two transactions with the same time, it picks one uniformly at random. $\maxbid(\bidtuple)$ outputs the maximum price bid in $\bidtuple$; given multiple bids at the same price, it outputs the first in temporal order. $\firsttime(\delayed{\pindex})$ outputs the earliest time from a list of times; $\popfirst(\delayed{\pindex})$ pops and outputs the earliest time. Although not specified here, bids must meet the validity rules -- i.e. must meet the minimum start bid $s$, minimum raise $\iota$ and must be staggered by a minimum time $\interval$. }
\label{fig:exec}
\end{figure*}

\begin{proof}[Proof of Observation \ref{obs:obs1}]
$\strategy_1$ counterbids after time $\delta$. Since $\blockinterval$ has distribution ${\sf Exp}(\lambda)$, the probability that any bid by $\strategy_0$ will have priority in a mined block is $\Pr[X < \interval]$, for $X$ a random variable with distribution ${\sf Exp}(\lambda)$. This holds with probability $1 - e^{-\lambda \interval}$. Thus, the payoff of $\strategy_1$ is positive if:

\begin{equation*}
    e^{-\lambda \interval} \payoffa - (1 - e^{-\lambda \interval})\bidcost > 0.
\end{equation*}
The observation follows.
\end{proof}

\begin{proof}[Proof of Observation \ref{obs:amplification}]
Since $\latency_i > \interval$, $\strategy_0$ can schedule its bids at intervals less than $\latency_i$. Thus, by the time $\player_1$ observes a bid $\bid_i$, $\player_0$ has already cast a new bid $\bid_{i+1}$. Since $\strategy_1$ is pure, $\player_1$ can choose a value $\bid_{i+1}$ that outbids $\strategy_0's$ reactive bid to $\bid_i$. Thus $\player_1$'s bids will always maintain the highest bid and win the PGA. 
 \end{proof}

\begin{observation}
\label{obs:max}
There exists a finite number of intervals $i_{max}$ for a cooperative PGA after which bidding is no longer profitable, even if the block has not yet been mined.
\end{observation}

\begin{proof}

Let $i_{max}$ be the greatest value for which the following inequality holds:

$$s \times (1 + \iota)^{i_{max}} \leq 1 + c, $$

\noindent where where $s$ is the starting bid and $n$ is the minimum bid increment as defined in Section \ref{sec:model}.

We can equivalently define a latest end time for the PGA $t_{end}=V[i_{max}]$.
\end{proof}

\noindent where $V[i] < t' < V[i+1]$.
%\ari{Equations should be integrated into prose using appropriate punctionation. I edited accordingsly.}\ari{{\bf Stuff still seems to be in flux, so I'll leave off commenting here.}}

\begin{observation}
\label{obs:ptime}
Denote by $p_{time}(t)$ the probability that a PGA has ended by some time $t$. For proof-of-work blockchains, $p_{time}(t) = 1- e^{-\lambda t}$.
\end{observation}

\begin{proof} %[Proof of Observation \ref{obs:ptime}]
Since the block interval in a proof-of-work blockchain is exponentially distributed, we model the PGA end time as an exponential distribution with rate parameter $\lambda$. The probability that the auction has ended by some positive time $t$ is given by the CDF of the exponential distribution:$ 1- e^{-\lambda t}.$
\end{proof}

%\ari{From this point on, theorem statements and proofs (in the appendix) would seem better than prose. Things become a bit hard to read...}

\begin{observation}
Assume that miners continuously update their blocks with the most profitable set of transactions. For players $\player_0, \player_1$ adhering to the cooperative strategy with parameters $D,W$, the probability that $\player_b$ will win the PGA is given by

$$\sum_{i=2k+b}^{i_{max}} (1 - p_{time}(V[i])) \times p_{time}(V[i+1]),$$

for $k\in [0,\frac{i_{end}}{2}]$.
\end{observation}

\begin{proof}
Notice that probability that a player wins a PGA is exactly the probability that the PGA terminates during an interval controlled by that player. This follows since miners continuously update their blocks with new more profitable transactions and the player in control of the interval will have the highest bid during its interval. 

We can define the probability that the game ends during a given interval as

$$p_{interval}(i) = (1 - p_{time}(V[i])) \times p_{time}(V[i+1]),$$

\noindent and the probability of $\player_b$ winning is the summation of $p_{interval}$ for all intervals that it controls.
\end{proof}

%For simplicity, we will assume uniform intervals, although the analysis generalized in a straightforward manner.

For any interval $i < i_{max}$, we can now define the expected payoff for each player during that interval, conditioned on the fact that the auction has already reached the interval $j\leq i$ and has not yet terminated. For each interval, we give two expected values: one for the ``bidder" --- i.e., the player whose turn it is to bid during that interval and one for the ``non-bidder". The bidder's expected reward is given by:
%As before $d_u$ is the uniform interval length.

$$E_{bidder}(i,j) = p_{interval}(i-j) \times (1-W[i]).$$

%$$E_{bidder}(i,j) = p_{interval}(i-j) \times \texttt{Max}(-c, 1-W[i])$$

%The maximum is taken to reflect the fact that players will only continue bidding while it is more profitable than backing out. 

The non-bidder's expected reward during a given interval is given by:

$$E_{non-bidder}(i,j)=p_{interval}(i-j) \times -c.$$

% \begin{observation}
% \label{obs:optimalbids}
% In a cooperative PGA, for all $i \leq i_{max}$, the optimal choice for bids is $W[0] = s$ and $W[i] =W[i-1]\times (1+\iota)$.
% \end{observation}

Now, we define the expected payoff for each player by continuing to follow the cooperative strategy at any point in the game. As in our definition of the strategy, we assume without loss of generality that $\player_0$ bids first and $\player_1$ bids second.

\begin{observation}
\label{obs:cooperate}
During the $j$th interval, the expected payoff for $\player_b$ to continue to cooperate for the rest of the auction is given by:

\begin{align*}
E_{cooperate}^{(\player_b)}(j) = &\sum_{i=0, \infty} E_{bidder}(2i+j + b,j) + \\  &\sum_{i=0, \infty} E_{non-bidder}(2i + j + (1-b),j) + \\
&(1-p_{time}(V[i_{end}])) \times -c
\end{align*}

\end{observation}

\begin{proof}
This follow directly form the previous observation. Previously, we defined the expected value of a bidder and a non-bidder in any interval. If we take an infinite sum of these expected values, alternating between bidder and non-bidder corresponding to the intervals in which the respective players bid, then we will get the expected payoff for the remainder of the game.
\end{proof}

% $$$$

% and symmetrically

% $$E_{cooperate}^{(\player_1)}(j) = \sum_{i=0, \infty} E_{non-bidder}(2i+j+1,j) +  \sum_{i=0, \infty} E_{bidder}(2i+j,j)$$

Next, we show the cost of deviation. For the grim-trigger cooperative strategy, each player knows that if they defect, the other player will raise the bid to make the game unprofitable -- i.e., each player's payoff will be $\lossmult(\bid_{{\sf lose}}$. 

While intuitively this arrangement (assuming that the other player will make good on its threat) makes deviation unwise, deviation may still be profitable if the auction ends before the other player can respond. This can happen due to the other player's latency -- i.e., if there is some time delay in which the deviation goes undetected. Additionally, this can occur due to the rate limit -- i.e., the deviation was detected, but the other player is not yet able to place a bid to make good on its threat. For the sake of this analysis, the latency and rate limit serve identical functions -- namely, to delay the other players response. Thus for simplicity, we define

$$\textsf{delay}^{(\player_b)} = \texttt{Max}[\latency_{1-b}, \delta].$$

\begin{observation}
\label{obs:deviate}
If for all $i$, it holds that $\textsf{delay}^{(\player_b)} < W[i+1] - W[i]$, then under the assumption that $\lossmult(\bid_{lose}) = -c$, the expected payoff for $\player_{1-b}$ deviating during the $i$th interval is given by

\begin{align*}
    E_{deviate}^{(\player_b)}(\textsf{delay}^{(\player_b)},i)&= p_{time}(\textsf{delay}^{(\player_b)}) \times \texttt{Max}(-c,1-W[i+1])\\ &+ (1- p_{time}(\textsf{delay}^{(\player_b)}) \times -c
\end{align*}.

\end{observation}

\begin{proof}
Since $\textsf{delay}^{(\player_b)} < W[i+1] - W[i]$, then $\player_b$ will notice and be able to react to $\player_{1-b}$'s deviation during the current (before it issues any other bids). At this point, $\player_b$ will react by bidding $1+c$, making the PGA unprofitable for both players.

Thus, $\player_{1-b}$'s expected profit is the expected reward that it makes during the time before the deviation is noticed, and $-c$ if the auction does not end during that period, which is given by the above equation.
\end{proof}

% \paragraph*{Uniform bid intervals}
% Above, we assumed two sets of values, $V$ and $W$, for the bid times and bid values, respectively. We now consider the simplified case where the intervals are all of uniform length. In particular, we have a constant $d_u$ such that for all $i$, ($V[i+1] - V[i]) = d_u$. Moreover, we assume that optimal values are chosen for $W[i]$ as shown in Observation \ref{obs:optimalbids}. 

% %$$W[i] = s \times (1 + n) ^ {i-1}$$

% We now have:

% $$p_{interval}(i) = (1 - p_{time}(d_u))^{i-1} \times p_{time}(d_u)$$

\begin{lemma}
\label{lem:nash}
Let $\player_{bidder}$ and $\player_{non-bidder}$ respectively denote the player who is and is not bidding during an interval according to the cooperative strategy. The grim-trigger cooperative strategy with parameters $D,W$ will yield a Nash Equilibrium if $\forall j$,

$$E^{(\player_{non-bidder})}_{cooperate}(j) \geq E^{(\player_{non-bidder})}_{deviate}(\textsf{delay}^{(\player_{bidder}),j}).$$
\end{lemma}

\begin{proof}
In order for a Nash equilibrium to hold for the grim-trigger cooperative strategy, it must be the case that for all points in time, it is more profitable to follow the cooperative agreement than it is to defect. 

Notice that deviating is most profitable at the beginning of one's interval since deviating later on will yield at most the same payoff, but conditioned on the game not having ended previously. Thus, in order for there to be an equilibrium, it must hold that $\forall j, \forall b \in \{0,1\}$, 

$$E^{(\player_b)}_{cooperate}(j) \geq E^{(\player_b)}_{deviate}(\textsf{delay}^{(\player_b),j}).$$

\noindent Notice that assuming the other player is cooperating, deviating is only profitable in intervals where one is the non-bidder. This is clear as when one already has the top bid, bidding again would decrease their own profit. Thus in order for the an equilibrium to hold, it must be true that the non-bidder.

\end{proof}

%Also notice that deviation only makes sense for the player who is not currently the highest bidder \steven{not strictly true since if you expect that the other player deviated but you haven't yet heard about it, you may want to deviate during your own interval}. We thus only need to consider deviations by the non-bidder. In order for an equilibrium to hold then, at every interval $i$, in which $\player_b$ is the non-bidder, the following must hold:

%\ari{It seems to me that you've broken a theorem (Observation 11) and its proof down into a sequence of observations and proofs. This seems fine to me, but I'd call them lemmas and put them in the appendix, while upgrading Observation 11 to a theorem---or at least a lemma.}

\section{OO Fees and Other System Designs}

We now discuss high-level consequences of OO fees in the context of different types of blockchain systems being built in academia and industry.  The complete enumeration of their impact is left to future work.

\begin{itemize}
    \item \textbf{Proof-of-stake systems} In any proof of stake system where forks are allowed (e.g.~\cite{daian2017snow}~\cite{peercoin}~\cite{kiayias2017ouroboros}), similar incentives as for Proof of Work exist for miner to orphan or rewrite history when profitable OO fees exist. Regardless of whether sealed-bid auctions (cf. Section~\ref{sec:whyrepeated}) are prominent, the next miner may overtake the previous block and replay its transactions in her block in an optimized order. Some stake-based protocols, on the other hand, attempt to provide a notion of \emph{finality}, and feature clients which will not revert history regardless of evidence presented. For systems in which the finality notion refers to a checkpoint that is done to cement multiple blocks that have already been created (e.g.~\cite{buterin2017casper}), all of the concerns remain the same, because rapid bidding wars affect the appeal of short-term (in particular, single block) forks. In blockchains that seek to finalize every block (e.g.~\cite{gilad2017algorand,kwon2014tendermint}), the potential of high profit OO fees implies that the honest (super-)majority assumption of such systems may not coincide with rational behavior. 
%In any proof of stake system where forks are allowed (e.g.~\cite{snowwhite}~\cite{peercoin}~\cite{cardano}), similar incentives as for Proof of Work exist for stakers to orphan or rewrite history when profitable OO fees exist.  Some stake-based protocols, on the other hand, attempt to provide a notion of \emph{finality}~\cite{finality}, and feature clients which will not revert history regardless of evidence presented.  In such systems, OO fees can still lead to short-term instability as described in~\cite{princetonpaper}, as there is no incentive to accept a history in the short term that is not maximally profitable to a miner given substantial ordering fees.  The effect of finality on long-range ordering attacks to claim OO fees is a subject we leave to future work.
    \item \textbf{Permissioned blockchain systems} Permissioned blockchains are currently being explored by many large financial institutions, often in use-cases like exchanges.  The importance of transaction order in these decentralized exchanges poses a variety of interesting questions for such systems.  For example, it is classically claimed that blockchains add auditability to existing workflows, but it is impossible to audit or objectively divine the real order in which a block producer received transactions on an asynchronous network.  Furthermore, because OO fees exist even on permissioned chains where on-chain assets are being exchanged, block producers in such chains must be chosen carefully and trusted with correct operation and ordering.
    \item \textbf{Sharded blockchain systems} Several sharded blockchain systems have been explored in both academia (e.g.~\cite{kokoris2018omniledger}) and industry (e.g.~\cite{buterin2017ethereum}).  One important consequence of sharding blockchains is the reduced security of each shard over the security of the whole system.  Generally, the effects of this reduced security are mitigated through random sampling from a large pool of potential validators, so an adversary would require substantial control of the pool to have a high chance of adversarially controlling a shard.  Unfortunately, as OO fees show, the security needs of shards may not be homogeneuous; a shard that operates a large decentralized exchange must pay miners higher rewards to ensure stability than a shard without such an exchange by the analysis in~\cite{carlsten2016instability}.  Because few sharded based systems are running in practice, substantial future work remains to fully enumerate potential attacks and their mitigations in sharded systems.
    \item \textbf{Other exchange designs}  In other blockchain-based exchange designs, it may still be possible for miners of the underlying system to manipulate prices, either in real time or retroactively.  This sets up similar incentives as those explored in this work.  One example is for channel based networks with public watchtowers such as~\cite{mccorry2018pisa}; miners can simply participate in exchanges, collecting old states through public watchtowers, and can selectively publish profitable states in a history rewriting attack that works as above.  Thus, it is not necessarily the case that the abstraction achieved by Layer 2 exchange systems is sufficient to prevent ordering attacks.
\end{itemize}
\end{document}